\documentclass[acmsmall]{acmart}

\AtBeginDocument{%
  \providecommand\BibTeX{{%
    \normalfont B\kern-0.5em{\scshape i\kern-0.25em b}\kern-0.8em\TeX}}}

\setcopyright{acmcopyright}
\copyrightyear{2018}
\acmYear{2018}
\acmDOI{10.1145/1122445.1122456}

\acmJournal{TOCL}
\acmVolume{37}
\acmNumber{4}
\acmArticle{111}
\acmMonth{8}

\usepackage{bussproofs, amsmath} 
\usepackage{color}
\usepackage{multicol}
\usepackage[all]{xy}
\usepackage{pgf, tikz, color}
\usepackage{changepage}
\usepackage{relsize}
\usetikzlibrary{decorations.pathreplacing}
\usepackage[utf8]{inputenc}
\usepackage{enumerate}

\newcommand{\ques}{\langle ? \rangle}

\newcommand{\Kt}{\textsf{Kt}}
\newcommand{\K}{\textsf{K}}

\DeclareSymbolFont{extraup}{U}{zavm}{m}{n}
\newcommand{\circw}{\circ\{}
\newcommand{\circb}{\bullet\{}
\newcommand{\LKt}{\textsf{G3Kt}} 
\newcommand{\bldia}{\Diamondblack}
\renewcommand{\blacklozenge}{\Diamondblack}
\renewcommand{\Box}{\square}

\newcommand{\disr}{(\lor)}
\newcommand{\conr}{(\wedge)}
\newcommand{\boxr}{(\Box)}
\newcommand{\diar}{(\Diamond)}
\newcommand{\id}{(\textsf{id})}
\newcommand{\wk}{(\textsf{w} )}\newcommand{\ctr}{(\textsf{c})}
\newcommand{\cut}{(\textsf{cut})}
\newcommand{\rf}{(\textsf{rf})}
\newcommand{\rp}{(\textsf{rp})}
\newcommand{\bboxr}{(\blacksquare)}
\newcommand{\bdiar}{(\bldia)}
\newcommand{\R}{\mathcal{R}}
\newcommand{\Q}{\mathcal{Q}}

\newcommand{\dprop}{(\textsf{dp})}

\newcommand{\SKT}{\textsf{SKT}}

\newcommand{\DKT}{\textsf{DKT}}
\newcommand{\bcirc}{\bullet}
\newcommand{\bdiam}{\bldia}
\newcommand{\diam}{\mathbin{\Diamond}}
\newcommand{\bBox}{\blacksquare}
\newcommand{\heart}{\ensuremath\heartsuit}

\newcommand{\dUTG}{\mathfrak{L}}
\newcommand{\UTGd}{\mathfrak{N}}

\newcommand{\lcut}{[}
\newcommand{\rcut}{]}

\newcommand{\negp}{(\mathsf{Path})}
\newcommand{\propr}{(\mathsf{Prop})}
\newcommand{\prop}{\textit{LabPr(P)}}
\newcommand{\propl}{\textit{LabPr(P)}}
\newcommand{\propdeep}{\textit{DeepPr(P)}}
\newcommand{\dgp}{(\mathsf{GP})}
\newcommand{\GP}{GP}
\newcommand{\GPd}{\textit{NestSt(GP)}}

\newcommand{\lgp}{(\mathsf{GP})}
\newcommand{\GPl}{\textit{LabSt(GP)}}

\newcommand{\dispath}{(\mathsf{Path})}
\newcommand{\lp}{(\mathsf{Path})}

\newcommand{\Pl}{\textit{LabSt(P)}}

\newcommand{\Pd}{\textit{NestSt(P)}}
\newcommand{\Pset}{P}

\newcommand{\display}{shallow nested }
\newcommand{\sub}{[\mathbf{s}]}
\newcommand{\dia}[1]{\langle #1 \rangle}
\newcommand{\empstr}{\varepsilon}
\newcommand{\der}{\mathcal{D}}

\AtEndPreamble{
\newtheorem{remark}[theorem]{Remark}
}

\begin{document}

\title{Display to Labeled Proofs and Back Again for Tense Logics}

\author{Agata Ciabattoni}
\email{agata@logic.at}
\orcid{0000-0001-6947-8772}
\affiliation{%
\institution{Technische Universit\"at Wien}
\department{Institut f\"ur Logic and Computation}
\streetaddress{Favoritensta{\ss}e 9-11}
\state{Wien}
\postcode{1040}
\country{Austria}}

\author{Tim S. Lyon}
\email{timothy_stephen.lyon@tu-dresden.de}
\orcid{0000-0003-3214-0828}
\affiliation{%
\institution{Technische Universit\"at Dresden}
\department{Institute of Artificial Intelligence}
\streetaddress{N\"othnitzer Str. 46}
\state{Dresden}
\postcode{01069}
\country{Germany}}



\author{Revantha Ramanayake}
\email{d.r.s.ramanayake@rug.nl}
\orcid{0000-0002-7940-9065}
\affiliation{%
\institution{University of Groningen}
\department{Bernoulli Institute for Mathematics, Computer Science and Artificial Intelligence, and CogniGron (Groningen Cognitive Systems and Materials Center)}
\streetaddress{Nijenborgh 4}
\state{Groningen}
\postcode{NL-9747}
\country{Netherlands}}

\author{Alwen Tiu}
\email{alwen.tiu@anu.edu.au}
\orcid{0000-0002-2695-5636}
\affiliation{%
\institution{The Australian National University}
\department{College of Engineering and Computer Science}
\streetaddress{108 North Road}
\city{Canberra}
\state{ACT}
\postcode{2601}
\country{Australia}}

\renewcommand{\shortauthors}{Ciabattoni et al.}

\begin{abstract}
We introduce translations between display calculus proofs and labeled calculus proofs in the context of tense logics. First, we show that every derivation in the display calculus for the minimal tense logic $\Kt$ extended with general path axioms can be effectively transformed into a derivation in the corresponding labeled calculus. Concerning the converse translation, we show that for $\Kt$ extended with path axioms, every derivation in the corresponding labeled calculus can be put into a special form that is translatable to a derivation in the associated display calculus. A key insight in this converse translation is a canonical representation of display sequents as labeled polytrees. Labeled polytrees, which represent equivalence classes of display sequents modulo display postulates, also shed light on related correspondence results for tense logics.
\end{abstract}



\maketitle

\section{Introduction}

    A crucial question for any logic is if it possesses an analytic calculus. An analytic calculus consists of rules that decompose a formula of the logic in a stepwise manner, and can be exploited to prove certain metalogical properties as well as develop automated reasoning methods. Since its introduction in the 1930's, Gentzen's sequent calculus (and equivalently, the tableaux calculus) has been a preferred formalism for constructing analytic calculi due to its simplicity. Unfortunately, this simplicity is also an obstacle: the formalism is not expressive enough to present many logics of interest. In response, many proof-theoretic formalisms extending the syntactic elements of the sequent calculus have been introduced over the last 30 years. 
Of particular interest in this paper are the formalisms of the labeled calculus~\cite{Gab96,Sim94,Vig00}, nested calculus~\cite{Bul92,Kas94,MarStr14AIML}, and display calculus~\cite{Bel82,Kra96}. Each formalism extends the sequent calculus in a seemingly unique way, suggesting distinct strengths, weaknesses, and expressive powers. 
There are trade-offs in employing one formalism as opposed to another, motivating a study of the interrelationships between the current patchwork (see, e.g. \cite{Pog10}) of proof systems.

In this paper, we consider proof calculi for a special class of multi-modal logics: extensions of the \emph{minimal tense logic} $\Kt$ with \emph{general path axioms} $\Pi A \rightarrow \Sigma A$ ($\Pi , \Sigma \in \{\Diamond, \Diamondblack\}^{*}$). Tense logics incorporate modalities that reference what is true in successor ($\Diamond$) and predecessor states ($\Diamondblack$). Such logics are used to model temporal notions having to do with future and past states of affairs. This class of logics provides a good case study 
for our proof-theoretic investigations since it includes many interesting/well-known logics and possesses a diverse proof theory.

Numerous analytic proof calculi have been presented for extensions of $\Kt$ such as labeled calculi~\cite{Bor08,BorNeg09}, nested calculi~\cite{GorPosTiu11}, and display calculi~\cite{Kas94,Kra96,Wan98dml}.
Since the term \emph{nested sequent} has been used in the literature to refer to slightly different objects, this is a good time to clarify our terminology. In this paper:
\begin{description}
\item[Nested sequent:] Any term generated via the BNF grammar $X ::= \varepsilon \ | \ A \ | \ X, X \ | \ {\circ} \{ X \} \ | \ {\bullet} \{ X \}$ where~$A$ is a tense formula.\footnote{We use $\varepsilon$ as the empty string, which in this context denotes the empty sequent.} Note that this \textit{extends} the typical definition of a nested sequent in the proof theory literature for modal (rather than tense) logics that uses a single nesting operator (e.g., the grammar for traditional nested sequents is usually given by the following BNF grammar: $X ::= \varepsilon \ | \ A \ | \ X, X \ | \ [X]$).

\item[Shallow nested calculus] (used here \textit{interchangeably}\footnote{The alternative term \emph{shallow nested sequent}  for \emph{display calculus} is due to~\cite{GorPosTiu11} whose motivation was to contrast the shallow inference rules of the display calculus with a proof calculus that uses deep inference instead.}
 with {\bf display calculus}) A proof calculus built from nested sequents in the sense above, where \textit{display rules} are used to unpack (`display') a formula nested under~$\circ$ and~$\bullet$ to bring it to the top-level, where the inference rules operate. 

\item[Deep nested calculus:] A proof calculus built from nested sequents in the sense above where the display rules are dispensed with, and the inference rules can apply inside arbitrary nestings of~$\circ$ and~$\bullet$ (i.e. \textit{deep inference} is implemented).
\end{description}
Deep nested calculi are better suited than shallow nested calculi for proving e.g. decidability~\cite{Bru06,GorPosTiu11} and interpolation~\cite{LyoTiuGorClo19}, due to the absence of the hard-to-control display rules that expand the proof-search space. Both shallow and deep nested calculi are typically {\em internal} in the sense that each sequent in a proof can be interpreted as a formula of the logic, whereas labeled calculi often appear to be {\em external} in the sense that the sequents cannot generally be interpreted as a formula of the logic (and use a language that explicitly encodes the semantics).

An effective way to relate calculi is by defining {\em translations}, i.e. functions that algorithmically transform any proof in a calculus into a proof of the same formula in another calculus. A crucial feature of such functions is that the structural properties of the derivation are preserved in the translation. Such embeddings permit the transfer of certain proof theoretic results, thus alleviating the need for independent proofs in each system, e.g. \cite{Fit12,GorRam12AIML,LyoBer19}. Moreover, they shed light on the role of certain syntactic features in proof calculi, and on the general problem of characterizing the relationships between different syntactic and semantic presentations of a logic~\cite{Pim18}.

In \cite{CiaLyoRam18} we obtained translations from shallow nested calculi to labeled calculi for Scott-Lemmon axiomatic extensions ($\Diamondblack^{h}\Diamond^{i}A\rightarrow\Diamond^{j}\Diamondblack^{k}A$ with $h,i,j,k\in\mathbb{N}$) of~$\Kt$. 
This paper extends these results to a larger set of tense logics, and answers an open question posed in that paper regarding the existence of labeled to nested translations for extensions of $\Kt$. 

We first show how to translate derivations in shallow nested calculi into derivations in labeled calculi for all general path extensions of $\Kt$. The reverse translation---from labeled to shallow nested---employs more sophisticated techniques and is only obtained for \textit{path axiom}---$\Pi A \rightarrow \ques A$ ($\Pi \in \{\Diamond, \blacklozenge\}^{*}$ and $\ques \in \{\Diamond, \blacklozenge\}$)---extensions of $\Kt$. Our proof strategy, described in the following paragraphs, ensures that each labeled sequent occurring in the derivation of a theorem of a path extension of $\Kt$ is interpretable as a nested sequent. This permits a translation from labeled to shallow nested sequent proofs. This translation witnesses a relation between the relational semantics and algebraic semantics (see e.g.~\cite{BlaRijVen01,Gor98}) for tense logics: the labeled calculi are clearly underpinned by the relational semantics; the shallow nested calculi, on the other hand, employ display rules that encode the algebraic residuation property between~$\blacklozenge$ (and $\Diamond$) in the antecedent and~$\Box$ (and $\blacksquare$, resp.) in the succedent of an implication. Indeed, the display rules have \textit{no} analog in the labeled calculi since the premise and conclusion translate to the same labeled sequent (see Lemma~\ref{isoresid}).

The ability to display any formula nested under structural connectives using the display rules is a crucial part in Belnap's~\cite{Bel82} proof of cut-elimination for arbitrary display calculi. However, the display rules greatly expand the proof search space, in particular when these rules interact with other structural rules (e.g. contraction) or structural rules that capture the modal/tense axioms of the formalized logic. In~\cite{GorPosTiu11}, the authors show how to translate display calculi to deep nested calculi, eliminating the display rules by employing deep inference. In our translation from display calculi to labeled calculi, display rules are not translated to inference rules; rather, they are dealt with by changing the representation of the nested sequent.
The key idea is that a nested sequent can naturally be interpreted as a labeled sequent whose binary relation between labels forms a {\em polytree} (i.e. a directed graph whose underlying undirected graph is a tree). 
 The polytree interpretation of a nested sequent has the crucial property of being invariant under display rules---applications of display rules to a nested sequent do not change its labeled polytree translation. 
Thus, display-equivalent nested sequents have a canonical representation as a labeled polytree sequent.
This representation also sheds light on the correspondence results between shallow and deep nested calculi for tense logics \cite{GorPosTiu11}. In particular, we show that the admissibility of display rules is independent from the admissibility of structural rules capturing the path axioms in tense logics, something that was
not observed in~\cite{GorPosTiu11}. 
This polytree representation also significantly simplifies the proof of Craig interpolation for the class of path extensions of~$\Kt$~\cite{LyoTiuGorClo19}.

Given that labeled polytree sequents correspond closely to nested sequents, one strategy to translate a labeled calculus to a shallow nested calculus is to translate a subset of the labeled calculus where all sequents are polytree sequents, and then show that the latter is complete, i.e. that it proves the same set of theorems as the unrestricted labeled calculi. One issue with this approach is that the property of being a polytree is not closed under some structural rules in labeled calculi, i.e. there could be instances of a rule where one of the premises is not a polytree but the conclusion is. To get around this issue, when translating from labeled to shallow nested, we first put our given derivation into a special form that makes use of so-called \emph{propagation rules}~\cite{CasCerGasHer97,Sim94,GorPosTiu11,LyoBer19}. Such rules allow us to eliminate certain structural rules from our labeled calculi and their derivations; this results in an \emph{internal} or \emph{refined} variant of the labeled calculus that---interestingly---inherits the nice properties of the original \emph{external} calculus. This methodology of eliminating structural rules to obtain refined calculi is of practical value in its own right~\cite{Lyo21}. In this paper, the methodology is used to provide a translation from labeled to shallow nested; however, this method is also useful in that it yields calculi suitable for proof-search and proving interpolation~\cite{LyoBer19,LyoTiuGorClo19}. Furthermore, this new form of the derivation permits an \emph{effective} (i.e. algorithmic) translation into a derivation of a deep nested calculus, that is to say, a step-by-step procedure defined at the level of the proof rules in the derivation; methods from~\cite{GorPosTiu11} are then applied to the deep nested derivation to further translate the proof into a proof of the corresponding shallow nested calculus. We note that this entire translation process is effective in the above sense, showing that the output derivation may be obtained algorithmically from local transformation on the input derivation. Our proof of admissibility of structural rules, in favor of propagation rules, for path axioms follows a similar methodology to that used in \cite{GorPosTiu11}, with one notable difference: in~\cite{GorPosTiu11}, the admissibility of display rules needs to be proved for {\em every} extension with path axioms, whereas in our case, admissibility of display rules is independent of the extensions, since the polytree representation makes the display rules superfluous. Our result thus suggests that perhaps display rules should be viewed as structural properties of sequents rather than as structural properties of proofs. This is analogous to, for example, internalizing the exchange rule as a property of sequents (i.e. commutativity and associativity of \emph{comma} in the sequent). 

The paper is structured as follows: Section 2 introduces the class of tense logics considered along with their associated shallow nested, labeled, and deep nested calculi. Section 3 presents labeled polytrees which are used to give the translation from nested proof systems to labeled proof systems as well as the reverse. In Section 4, we provide an effective translation from shallow nested proofs to labeled proofs for all \emph{general path} extensions of $\Kt$. Section 5 gives the reverse translation from labeled proofs to shallow nested proofs for \emph{path} extensions of $\Kt$. Section 6 discusses consequences and potential applications.

We summarize below the calculi considered in this paper and illustrate the effective proof-transformations (which transform the shape of a derivation and preserve the language of the calculus; indicated by a dotted arrow) and translations (which not only transform the shape of the derivation, but translate the language of the calculus; indicated by solid arrow) obtained in this paper.


\subsubsection*{Base Calculi and Extensions ($GP$ general path axioms, $P$ path axioms):}

\begin{center}
\begin{tabular}{|c|c|c|c|c|}
\hline
Base Calc. & Type & Gen. Path Str. Rules & Path Str. Rules & Propagation Rules \\
\hline
$\LKt$~\cite{Bor08,CiaLyoRam18} & labeled & $\GPl$ & $\Pl$ & $\propl$\\
\hline
$\SKT$~\cite{GorPosTiu11} & Shal. Nes. & $\GPd$ & $\Pd$ &  \\
\hline
$\DKT$~\cite{GorPosTiu11} & Deep Nes. &  &  & $\propdeep$\\
\hline

\end{tabular}
\end{center}

\subsubsection*{Effective Transformations/Translations:}

\begin{center}
\xymatrix{
	& \LKt + \GPl	&      \LKt + \Pl\ar@{.>}[rr]^{Lem.~\ref{internal}}\ar@/^-.8pc/@{->}[d]_{Thm.~\ref{reverse-translation-theorem}} &  & \LKt + \propl  \ar@{->}[d]^{Lem.~\ref{G3Kt-Path-to-DKT-Path}}  & 		\\
	& \SKT + \GPd\ar@{->}[u]^{Thm.~\ref{DisToLab}} 	&  \SKT + \Pd\ar@/^-.8pc/@{->}[u]_{Thm.~\ref{DisToLab}}    &  & \DKT + \propdeep\ar@{.>}[ll]^{Lem.~\ref{sktdktequiv}}	  &
}
\end{center}





\section{Nested and Labeled Calculi for Tense Logics}

For convenience, we take the language $\mathcal{L}_{\Kt}$ as consisting of formulae in negation normal form. In particular, formulae are built from the literals~$p$ and $\overline{p}$ using the $\wedge$, $\vee$, $\Diamond$, $\Box$, $\bldia$, and $\blacksquare$ operators. Note that all results hold also for the full language where the $\neg$, $\rightarrow$, and $\leftrightarrow$ operators are taken as primitive. The language $\mathcal{L}_{\Kt}$ is explicitly defined via the following BNF grammar:
$$
A ::= p \ | \ \overline{p} \ | \ A \wedge A \ | \ A \vee A \ | \ \Box A \ | \ \Diamond A \ | \ \blacksquare A \ | \ \blacklozenge A
$$
For an introduction to tense logics and their accompanying Kripke semantics, we refer the reader to the following references:~\cite{BlaRijVen01,Gar13}.

Intuitively, we interpret $\Box A$ as claiming that the formula $A$ holds at every point in the immediate future, whereas $\blacksquare A$ is interpreted as claiming that $A$ holds at every point in the immediate past. Similarly, we interpret the formula $\diam A$ as claiming that $A$ holds at some point in the immediate future, while $\bdiam A$ intuitively means that $A$ holds at some point in the immediate past.

Define~$\overline{A}$ inductively as follows.
\begin{center}
\begin{tabular}{p{0.4\textwidth} @{\hskip 6em} p{0.4\textwidth}}
\begin{itemize}

\item[(1)] If $A = p$, then $\overline{A} = \overline{p}$;

\item[(2)] If $A = \overline{p}$, then $\overline{A} = p$;

\item[(3)] If $A = B \wedge C$, then $\overline{A} = \overline{B} \vee \overline{C}$;

\item[(4)] If $A = B \vee C$, then $\overline{A} = \overline{B} \wedge \overline{C}$;

\end{itemize}

&

\begin{itemize}

\item[(5)] If $A = \Box B$, then $\overline{A} = \Diamond \overline{B}$;

\item[(6)] If $A = \Diamond B$, then $\overline{A} = \Box \overline{B}$;

\item[(7)] If $A = \blacksquare B$, then $\overline{A} = \blacklozenge \overline{B}$;

\item[(8)] If $A = \blacklozenge B$, then $\overline{A} = \blacksquare \overline{B}$.

\end{itemize}
\end{tabular}
\end{center}
We define the negation $\neg A$ of formula~$A$ as $\overline{A}$, the conditional $A \rightarrow B$ as $\overline{A} \vee B$, and the biconditional $A \leftrightarrow B$ as $A \rightarrow B \wedge B \rightarrow A$.

The tense logic~$\Kt$---a conservative extension of the normal modal logic~$\textsf{K}$---is typically axiomatized as shown below (see, e.g.~\cite{BlaRijVen01,ChaZak92}). 


\begin{center}
\begin{tabular}{c @{\hskip 1em} c @{\hskip 1em} c}

$A \rightarrow (B \rightarrow A)$

&

$(\neg B \rightarrow \neg A) \rightarrow (A \rightarrow B)$

&

$(A \rightarrow (B \rightarrow C)) \rightarrow ((A \rightarrow B) \rightarrow (A \rightarrow C))$

\end{tabular}
\end{center}

\begin{center}
\begin{tabular}{c @{\hskip 2em} c @{\hskip 2em} c @{\hskip 2em} c @{\hskip 2em} c @{\hskip 2em} c}

$A \rightarrow \Box \blacklozenge A$

&

$A \rightarrow \blacksquare \Diamond A$

&

$\Box A \leftrightarrow \neg \Diamond \neg A$

&

$\blacksquare A \leftrightarrow \neg \blacklozenge\neg A$

&

\AxiomC{$A$}
\UnaryInfC{$\Box A$}
\DisplayProof

&

\AxiomC{$A$}
\UnaryInfC{$\blacksquare A$}
\DisplayProof

\end{tabular}
\end{center}

\begin{center}
\begin{tabular}{c @{\hskip 3em} c @{\hskip 3em} c}

$\Box (A \rightarrow B) \rightarrow (\Box A \rightarrow \Box B)$

&

$\blacksquare (A \rightarrow B) \rightarrow (\blacksquare A \rightarrow \blacksquare B)$

&

\AxiomC{$A$}
\AxiomC{$A \rightarrow B$}
\BinaryInfC{$B$}
\DisplayProof

\end{tabular}
\end{center}

As mentioned previously, the logics we consider in this paper are extensions of $\Kt$ with \emph{general path axioms} of the form $\ques_{1} ... \ques_{n} A \rightarrow \ques_{n+1} ... \ques_{n+m} A$ where each~$\ques_{j}$ is either~$\Diamond$ or~$\Diamondblack$. Occasionally, we may use $\dia{F}$, $\dia{G}$, $\ldots$ to represent either a~$\Diamond$ or a~$\Diamondblack$. Also, note that when $n = 0$, the antecedent of the path axiom is free of diamonds (i.e. it is of the form $A \rightarrow \ques_{1} ... \ques_{m} A$), and when $m = 0$, the consequent is free of diamonds (i.e. it is of the form $\ques_{1} ... \ques_{n} A \rightarrow A$). We will use the notation $\Pi A \rightarrow \Sigma A$ to represent such axioms.
This class of axioms contains many well-known axioms such as reflexivity $A \rightarrow \Diamond A$, confluence $\blacklozenge \Diamond A \rightarrow \Diamond \blacklozenge A$, and partial-functionality $\blacklozenge \Diamond A \rightarrow A$. We will use $\GP$ to denote an arbitrary set of general path axioms and write $\Kt + \GP$ to mean the minimal tense logic $\Kt$ extended with the axioms from $\GP$; note that this notation extends straightforwardly to any set $S$ of formulae, i.e. $\Kt + S$ will be used to represent extensions of $\Kt$ with the formulae from $S$, as well as the corresponding logic (i.e. the set of theorems). Last, we let $\Kt + S \vdash A$ denote that~$A$ is a theorem of the logic~$\Kt+S$.

\textit{Path axioms} are general path axioms where the consequent of the axiom is restricted to a single-diamond formula, \textit{i.e.} any formula of the form $\ques_{1} ... \ques_{n} A \rightarrow \ques_{n+1} A$ is a path axiom. We focus on this class of axioms because the translation methods presented in this paper only allow us to translate derivations from labeled to nested for the logics $\Kt + P$, where $P$ is an arbitrary set of path axioms. Nevertheless, this class of axioms still contains well-known axioms such as transitivity $\Diamond \Diamond A \rightarrow \Diamond A$, symmetry $\blacklozenge A \rightarrow \Diamond A$, and Euclideanity $\blacklozenge \Diamond A \rightarrow \Diamond A$.
 

\subsection{Shallow Nested (Display) Calculi for Tense Logics}\label{Shallow_Nested_Intro_Subsection}

We will present Gor\'{e} \textit{et al.}'s~\cite{GorPosTiu11} \display calculus~$\SKT$ for $\Kt$.
This calculus can be seen as a one-sided version of Kracht's~\cite{Kra96} display calculus for~$\Kt$, and also as a variant of Kashima's~\cite{Kas94} calculus. 

%

The shallow nested calculus is modular in the sense that certain axiomatic extensions of $\Kt$ can be captured by adding equivalent structural rules to $\SKT$. Moreover, $\SKT$ allows for a uniform proof of cut-elimination where cut is eliminable from any derivation of $\SKT$ extended with any number of \textit{substitution-closed linear structural rules} (see \cite{GorPosTiu11} for details). This makes the shallow nested calculus a good candidate for capturing large classes of tense logics in a unified, cut-free manner. 
The nested sequents of $\SKT$ are generated by the following grammar where $A$ is a tense formula in $\mathcal{L}_{\Kt}$. 
$$
X ::= \empstr \ | \ A \ | \ X, X \ | \ {\circ} \{ X \} \ | \ {\bullet} \{ X \}
$$
We assume comma to commute and associate, meaning, for example, that we may freely re-write a nested sequent of the form $X,Y,Z$ as $Z,X,Y$ when performing derivations in $\SKT$. Also, $\empstr$ represents the \emph{empty string} or \emph{empty sequent}, which acts as an identity element for comma (e.g. we identify $X,\empstr$ with $X$), and so, $\empstr$ will be implicit in nested sequents, but not explicitly appear.

A characteristic of nested sequents is that each can be translated into an equivalent formula in the language $\mathcal{L}_{\Kt}$, that is, each connective introduced in the language of nested sequents acts as a \emph{proxy} for a logical connective (cf.~\cite{Bel82,Kra96,GorPosTiu11}). The interpretation~$\mathcal{I}$ of a nested sequent as a tense formula is defined as follows:
\begin{center}
\begin{multicols}{3}
\begin{itemize}

\item[(1)] $\mathcal{I}(\varepsilon) = \top$

\item[(2)] $\mathcal{I}(A) = A \text{ for } A \in \mathcal{L}_{\Kt}$

\item[(3)] $\mathcal{I}(X, Y) = \mathcal{I}(X) \lor \mathcal{I}(Y)$

\item[(4)] $\mathcal{I}(\circ \{ X \}) = \Box \mathcal{I}(X)$

\item[(5)] $\mathcal{I}(\bcirc \{ X \}) = \bBox \mathcal{I}(X)$

\end{itemize}
\end{multicols}
\end{center}


It will occasionally be useful to refer to the \emph{substructures} of a nested sequent $X$. We say that a sequent $Y$ is a \emph{substructure} of $X$ if and only if $Y \in \mathfrak{S}(X)$, where the set of \emph{substructures} of $X$, written $\mathfrak{S}(X)$, is inductively defined as follows:
\begin{center}
\begin{multicols}{2}
\begin{itemize}

\item[(1)] $\mathfrak{S}(\varepsilon) = \emptyset$

\item[(2)] $\mathfrak{S}(A) = \{A\} \text{ for } A \in \mathcal{L}_{\Kt}$

\item[(3)] $\mathfrak{S}(X) = \{X\} \cup \mathfrak{S}(Y) \cup \mathfrak{S}(Z)$, if $X = Y,Z$



\item[(4)] $\mathfrak{S}(X) = \{X\} \cup \mathfrak{S}(Y)$, if $X = {\circ}\{Y\}$ or ${\bullet}\{Y\}$


\end{itemize}
\end{multicols}
\end{center}

\begin{definition}[The Calculus $\SKT$~\cite{GorPosTiu11}]\

\begin{center}
\begin{tabular}{c c c}
\AxiomC{} \RightLabel{$\id$}
\UnaryInfC{$X, p, \overline{p}$}
\DisplayProof

\hspace*{.25 cm}

&
\AxiomC{$X, A,B $}
\RightLabel{$\disr$}
\UnaryInfC{$X, A\lor B$}
\DisplayProof

\hspace*{.25 cm}

&
\AxiomC{$X, A$}
\AxiomC{$X, B$}
\RightLabel{$\conr$}
\BinaryInfC{$X, A\land B$}
\DisplayProof
\end{tabular}
\end{center}


\begin{center}
\begin{tabular}{cccc}
\AxiomC{$X, Y, Y$}
\RightLabel{$\ctr$}
\UnaryInfC{$X, Y$}
\DisplayProof

\hspace*{.25 cm}

&
\AxiomC{$X$}
\RightLabel{$\wk$}
\UnaryInfC{$X, Y$}
\DisplayProof

\hspace*{.25 cm}

&
\AxiomC{$X, \circw Y \}$}
\RightLabel{$\rf$}
\UnaryInfC{$\circb X \}, Y $}
\DisplayProof

\hspace*{.25 cm}

&
\AxiomC{$X, \circb Y \}$}
\RightLabel{$\rp$}
\UnaryInfC{$\circw X \}, Y$}
\DisplayProof
\end{tabular}
\end{center}


\begin{center}
\begin{tabular}{c c c c}
\AxiomC{$X, \circb A \}$}
\RightLabel{$\bboxr$}
\UnaryInfC{$X, \blacksquare A$}
\DisplayProof

\hspace*{.25 cm}

&
\AxiomC{$X, \circw A \}$}
\RightLabel{$\boxr$}
\UnaryInfC{$X, \Box A$}
\DisplayProof
%

\hspace*{.25 cm}
&
\AxiomC{$X, \circb Y, A \}, \bldia A$}
\RightLabel{$\bdiar$}
\UnaryInfC{$X, \circb Y \}, \bldia A$}
\DisplayProof

\hspace*{.25 cm}

&
\AxiomC{$X, \circw Y, A \}, \Diamond A$}
\RightLabel{$\diar$}
\UnaryInfC{$X, \circw Y \}, \Diamond A$}
\DisplayProof
\end{tabular}
\end{center}

\end{definition}
$\SKT$ is referred to  as a \display sequent calculus
because (i)~the~$\circ\{ \cdot \}$ and $\bcirc\{ \cdot \}$ provide (two types of) nestings and (ii)~all the rules are shallow in the sense that they operate at the \textit{root} or \textit{top-level} of the sequent (i.e. rules are only applied to formulae or structures that do not occur within nestings).

\begin{definition}[Display Property]
A calculus has the \emph{display property} if it contains a set of rules (called \emph{display rules}) such that
 for any sequent~$X$ containing a substructure~$Y$, there exists a sequent~$Z$ such that~$Y,Z$ is derivable from~$X$ using only the display rules.
\end{definition}
The display property states that any substructure in~$X$ can be brought to the \emph{top level} using the display rules.
The calculus $\SKT$ has the display property when $\{\rp,\rf\}$ is chosen to be the set of display rules, i.e., the residuation rules $\rp$ and $\rf$ serve as the display rules in $\SKT$. A pair of nested sequents are \emph{display equivalent} when they are mutually derivable using only the display rules. The display property is significant since it is a crucial component in the proof of cut-elimination (see~\cite{Bel82}).

A modular method of obtaining a cut-free extension of the base calculus for~$\Kt$ by a large class of axioms inclusive of the general path axioms was introduced in~\cite{Kra96} (see also~\cite{CiaRam16}).
Following~\cite{Kra96}, we present the rule~$\dgp$ corresponding to a general path axiom $\ques_{1} ... \ques_{n} A \rightarrow \ques_{n+1} ... \ques_{n+m} A$:

\begin{center}
\AxiomC{$X, \star_{n+1} \{ ... \star_{n+m} \{ Y \} ... \}$}
\RightLabel{$\dgp$}
\UnaryInfC{$X, \star_{1} \{ ... \star_{n} \{ Y \} ... \}$}
\DisplayProof
\end{center}

\noindent
Here if~$\ques_{j}=\Diamond$ then $\star_{j}=\circ$, and if~$\ques_{j}=\Diamondblack$ then $\star_{j}=\bullet$.
%

Since path axioms form a proper subclass of the general path axioms, the rule $\dgp$ can be specialized to the rule $\dispath$ for any given path axiom $\ques_{1}... \ques_{n} A \rightarrow \ques_{n+1} A$:


\begin{center}
\AxiomC{$X, \star_{n+1} \{  Y \}$}
\RightLabel{$\dispath$}
\UnaryInfC{$X, \star_{1} \{ ... \star_{n} \{ Y \} ... \}$}
\DisplayProof
\end{center}


\begin{theorem}[\cite{GorPosTiu11}]\label{sktcutlim} The $\cut$ rule
\begin{center}
\AxiomC{$X,A$}
\AxiomC{$\overline{A},Y$}
\RightLabel{$\cut$}
\BinaryInfC{$X,Y$}
\DisplayProof
\end{center}
is admissible in $\SKT+\GPd$.
\end{theorem}

\begin{theorem}[\cite{GorPosTiu11}]\label{sktcompleteness} $\Kt+\GP \vdash A$ iff $A$ is derivable in $\SKT+\GPd$.
\end{theorem}

\subsection{Labeled Calculi for Tense Logics}
\label{sec:labeled-calculi}

Labeled sequents~\cite{Fit83,Min97} generalize 
Gentzen sequents by prefixing \textit{state variables} to formulae
occurring in the sequent and by making the relational semantics explicit in the syntax. Labeled sequents are defined via the BNF grammar below:
$$
\Lambda ::= \empstr \ | \ x : A \ | \ \Lambda, \Lambda \ | \ Rxy, \Lambda
$$
where $A \in \mathcal{L}_{\Kt}$, and $x$ and $y$ are among a denumerable set $x, y, z, \ldots$ of labels. We often write a labeled sequent $\Lambda$ as $\mathcal{R},\Gamma$ where $\mathcal{R}$ consists of the \emph{relational atoms} of the form $Rxy$ occurring in $\Lambda$ and $\Gamma$ consists of the \emph{labeled formulae} of the form $x : A$ occurring in $\Lambda$. Additionally, characters such as $\mathcal{R}, \mathcal{Q}, \ldots$ will be used to denote \emph{sets} of relational atoms and Greek letters such as $\Gamma, \Delta, \ldots$ will be used to denote \emph{multisets} of labeled formulae. As in the case of nested sequents, we assume that comma commutes and associates, meaning that each labeled sequent $\Lambda$ can indeed be written in the form above, and also assume that $\empstr$ represents the \emph{empty string} or \emph{empty sequent}, which acts as an identity element for comma and occurs only implicitly in labeled sequents.

A labeled sequent can be viewed as a directed graph (defined using~$\mathcal{R}$) with formulae decorating each node~\cite{CiaLyoRam18}. Note that in a labeled sequent $\Lambda = \R, \Gamma$ commas between relational atoms are interpreted conjunctively, the comma between $\R$ and $\Gamma$ is interpreted as an implication, and the commas between the labeled formulae in $\Gamma$ are interpreted disjunctively.

Vigan\`{o}~\cite{Vig00} constructed labeled sequent calculi for non-classical logics whose semantics are defined by Horn formulae. Negri~\cite{Neg05} extended the method to generate cut-free and contraction-free labeled sequent calculi for the large family of modal logics whose Kripke semantics are defined by geometric (first-order) formulae, similar to what had been achieved by Simpson for intuitionistic modal logics in his dissertation~\cite{Sim94}. The proof of cut-elimination is general in the sense that it applies uniformly to every modal logic defined by geometric formulae; this result has been extended to intermediate and other non-classical logics~\cite{Bor08,DycNeg12} and to arbitrary first-order formulae~\cite{DycNeg15}.

We begin by extending in the natural way the usual labeled sequent calculus for~$\textsf{K}$ to a labeled sequent calculus for~$\Kt$.

\begin{definition}[The labeled sequent calculus $\LKt$~\cite{Bor08,CiaLyoRam18}]\label{labeled_Calc_Def}\


\begin{center}
\begin{tabular}{c @{\hskip 1em} c @{\hskip 1em} c}
\AxiomC{} \RightLabel{$\id$}
\UnaryInfC{$\R, x: p, x: \overline{p}, \Gamma$}
\DisplayProof

&

\AxiomC{$\R, x:A,x:B, \Gamma $}
\RightLabel{$\disr$}
\UnaryInfC{$\R, x:A\lor B, \Gamma$}
\DisplayProof

&

\AxiomC{$\R, x:A, \Gamma$}
\AxiomC{$\R, x:B, \Gamma$}
\RightLabel{$\conr$}
\BinaryInfC{$\R, x:A\land B, \Gamma$}
\DisplayProof
\end{tabular}
\end{center}

\begin{center}
\begin{tabular}{c @{\hskip 1em} c}
\AxiomC{$\R, Rxy, y:A, \Gamma$}
\RightLabel{$\boxr^{*}$}
\UnaryInfC{$\R, x:\Box A, \Gamma$}
\DisplayProof

&

\AxiomC{$\R, Rxy, y:A, x:\Diamond A, \Gamma$}
\RightLabel{$\diar$}
\UnaryInfC{$\R, Rxy, x:\Diamond A, \Gamma$}
\DisplayProof

\end{tabular}
\end{center}

\begin{center}
\begin{tabular}{c @{\hskip 1em} c}
\AxiomC{$\R, Ryx, y:A, \Gamma$}
\RightLabel{$\bboxr^{*}$}
\UnaryInfC{$\R, x:\blacksquare A, \Gamma$}
\DisplayProof

&

\AxiomC{$\R, Ryx, y:A, x:\bldia A, \Gamma$}
\RightLabel{$\bdiar$}
\UnaryInfC{$\R, Ryx, x:\bldia A, \Gamma$}
\DisplayProof

\end{tabular}
\end{center}

\end{definition}

The~$\boxr$ and~$\bboxr$ rules have a side condition: ($\ast$) the variable~$y$ does not occur in the conclusion. When a variable is not allowed to occur in the conclusion of an inference, we refer to it as an \emph{eigenvariable}.

A general path axiom is a Sahlqvist formula, and hence it has a first-order frame correspondent which can be computed---even in the case of tense logics (see~\cite{BlaRijVen01}). Following the method in~\cite{Neg05}, the labeled structural rule~$\lgp$ corresponding to a general path axiom~$\Pi A \rightarrow \Sigma A$ is obtained below. We let $R_{\Diamond}xy := Rxy$, and $R_{\blacklozenge}xy := Ryx$, and use $R_{\ques}xy \in \{R_{\Diamond}xy, R_{\blacklozenge}xy\}$, that is, $R_{\ques}xy$ may represent either $Rxy$ or $Ryx$. Furthermore, we use the following notation: $\R_{\Pi}xy := R_{\ques_{1}}xy_{1},...,R_{\ques_{m}}y_{m}y$ for $\Pi = \ques_{1}...\ques_{m}$ and $\R_{\Sigma}xy := R_{\ques_{1}}xz_{1},...,R_{\ques_{n}}y_{n}y$ for $\Sigma = \ques_{1}...\ques_{n}$. 
\begin{center}
\AxiomC{$\R, \R_{\Pi}xy, \R_{\Sigma}xy, \Gamma$} \RightLabel{$\lambda GP$}
\RightLabel{$\lgp^{*}$}
\UnaryInfC{$\R, \R_{\Pi}xy, \Gamma$}
\DisplayProof
\end{center}

\noindent
This rule also has a side condition: ($\ast$) all variables occurring in the relational atoms $\R_{\Sigma}xy$ with the exception of $x$ and $y$ are eigenvariables.


\begin{remark}\label{rem}
In the rule above, some care is needed in the boundary cases when $\Pi$ or $\Sigma$ are empty strings of diamonds. The table below specifies the instances of the rule depending on whether the string is non-empty (marked with~$+$), or empty (marked with~$\epsilon$):
 
\begin{center}
\begin{tabular}{| c | c | c | c|}
\hline
$\Pi$ & $\Sigma$ & Premise & Conclusion    \\
\hline
$+$  & $+$  &
$\R, R_ {\Pi}xy, R_ {\Sigma}xy, \Gamma$
&
$\R, R_ {\Pi}xy, \Gamma$
\\
\hline
$+$  & $\epsilon$  &
$\R, R_ {\Pi}xy, x = y, \Gamma$
&
$\R, R_ {\Pi}xy, \Gamma$
\\
\hline
$\epsilon$  & $+$  &
$\R, R_ {\Sigma}xx, \Gamma$
&
$\R, \Gamma$
\\
\hline
$\epsilon$  & $\epsilon$  &
$\R, \Gamma$
&
$\R, \Gamma$
\\
\hline
\end{tabular}
\end{center}

Note that $[x/y]$ stands for the substitution of the label $y$ by the
label $x$. Also, when $\Pi = \epsilon$ or $\Sigma = \epsilon$, $R_{\Pi}xy$ and $R_{\Sigma}xy$ are taken to be $x=y$. For the second entry of the table, we must extend our labelled sequents to include equations of the form $x = y$ as relational atoms and extend our calculus with equality and substitution rules (see~\cite{Neg05}). All results of the paper continue to hold with respect to the addition of such rules; however, we omit their presentation for simplicity. For the third and fourth entries in the table, the equality symbols that arise have been eliminated through substitutions and suitable argumentation. This argumentation can be formalized using the equality and substitution rules specified in~\cite{Neg05}.
\end{remark}

Since particular attention will be paid to the class of path axioms (specifically in section \ref{Trans_Path_Ax_Ext_Subsection}), we also explicitly give the structural rule $\lp$ which is an instance of $\lgp$ and corresponds to a path axiom $\Pi A \rightarrow \ques A$:

\begin{center}
\AxiomC{$\R, R_{\Pi}xy, R_{\ques}xy, \Gamma$} \RightLabel{$\lp$}
\UnaryInfC{$\R, R_{\Pi}xy, \Gamma$}
\DisplayProof
\end{center}

We use the name $\GPl$ to represent the set of labeled structural rules corresponding to a set $\GP$ of general path axioms and the name $\Pl$ to refer to the set of labeled structural rules corresponding to a set $\Pset$ of path axioms.

It is straightforward to apply the arguments and methods concerning labeled calculi for modal and tense logics, presented in~\cite{Bor08,Neg05}, to conclude the following:

\begin{lemma}\label{g3ktstrucadmiss} Let $\R$ and $\Q$ be sets of relational atoms with $\R \cap \Q = \emptyset$, and $\Gamma$ and $\Delta$ be multisets of labeled formulae. The following rules are admissible in~$\LKt + \GPl$:
\begin{center}
\begin{tabular}{c c c c}
\AxiomC{$\R, \Gamma$}
\RightLabel{$\sub$}
\UnaryInfC{$\R[x/y], \Gamma[x/y]$}
\DisplayProof

&

\AxiomC{$\R, \Gamma, \Delta, \Delta$}
\RightLabel{$\ctr$}
\UnaryInfC{$\R, \Gamma, \Delta$}
\DisplayProof

&

\AxiomC{$\R, \Gamma$}
\RightLabel{$\wk$}
\UnaryInfC{$\R, \Q, \Gamma, \Delta$}
\DisplayProof

&

\AxiomC{$\R, \Gamma, x:A$}
\AxiomC{$\R, \Gamma, x:\overline{A}$}
\RightLabel{$\cut$}
\BinaryInfC{$\R, \Gamma$}
\DisplayProof
\end{tabular}
\end{center}
We assume that duplicates of relational atoms are removed from $\R[x/y]$ in the conclusion of the $\sub$ rule.
\end{lemma}

\begin{theorem} $\Kt + \GP \vdash A$ iff $x:A$ is derivable in $\LKt + \GPl$.

\end{theorem}

\subsection{Deep Nested Calculi for Tense Logics }

In this section we present Gor\'e \textit{et al.}'s \cite{GorPosTiu11} deep nested calculus $\DKT$ for $\Kt$, as well as extensions of $\DKT$ with inference rules---referred to as \emph{propagation rules}---that correspond to the class of \emph{path axioms}. 
%
Although we will show how to translate shallow nested derivations into labeled derivations for the logics $\Kt + \GP$, we consider path axioms here because the reverse translation from labeled proofs to shallow nested proofs is only known for the smaller class of logics $\Kt + \Pset$. The deep nested calculi presented here will be used to facilitate and simplify the reverse translation.

Our calculi make use of nested sequents from the same language as $\SKT$. Every nested sequent $X := Y, \circw Z_{1} \}, ..., \circw Z_{n} \}, \circb W_{1} \}, ..., \circb W_{m} \}$ ($Y$ contains no nesting) may be represented as a tree with two types of edges \cite{Kas94,GorPosTiu11}. The tree of $X$, denoted $tree(X)$, is shown below:
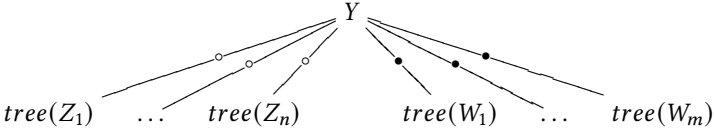
\begin{center}
\begin{tabular}{c}
\xymatrix@C=1em{
		& &   & Y\ar@{-}[dlll]|-{\circ}\ar@{-}[dll]|-{\circ}\ar@{-}[dl]|-{\circ}\ar@{-}[dr]|-{\bullet}\ar@{-}[drr]|-{\bullet}\ar@{-}[drrr]|-{\bullet} 	&   & 	&  & 		\\
	tree(Z_{1})	& \hdots  & tree(Z_{n}) &  	& tree(W_{1})	& \hdots &  tree(W_{m}) & 
}
\end{tabular}
\end{center} 
A nested sequent that contains holes in place of formulae is called a \emph{context}. Like nested sequents, contexts may be represented as trees, but where nodes are additionally labeled with holes. A context with a single hole is written as $X[]$ and a context with multiple holes is written as $X[]\cdots[]$. We may compose a context with sequents to obtain a sequent (e.g. $X[Y_{1}]\cdots[Y_{n}]$ is a sequent where $X[]\cdots[]$ is a multi-hole context and $Y_{1}$, ..., $Y_{n}$ are sequents); graphically, this corresponds to fusing the root of the tree of each sequent with the node in the context where the associated hole occurs. Note that this notation is the opposite of what is often used for nested sequent calculi for \emph{modal} logics in the literature, though is consistent with the notation used in the literature for nested sequent calculi for tense logics (cf.~\cite{GorPosTiu11}).

When representing a context graphically, each hole will label a unique node in the corresponding tree. For a single-hole context we write $X[]_{i}$ to indicate the node $i$ where the hole occurs, and for a multi-hole context we write $X[]_{i_{1}}\cdots[]_{i_{n}}$ to indicate the unique nodes in the tree that correspond to each hole.


\begin{definition}[The Calculus $\DKT$~\cite{GorPosTiu11}\footnote{As shown in~\cite{GorPosTiu11}, copying the principal formula in the $(\Box)$ and $(\blacksquare)$ rules is useful when performing proof-search, despite being unnecessary for completeness of the calculus. Still, we make use of the same rules here since we will leverage methods presented in~\cite{GorPosTiu11} that make use of the calculus $\DKT$ in the form above.}]\

\begin{center}
\begin{tabular}{c c c}
\AxiomC{} \RightLabel{$\id$}
\UnaryInfC{$X \lcut p, \overline{p} \rcut$}
\DisplayProof

&

\AxiomC{$X \lcut A, Y \rcut$}
\AxiomC{$X \lcut B, Y \rcut$}
\RightLabel{$(\land)$}
\BinaryInfC{$X \lcut A\land B, Y \rcut$}
\DisplayProof

&

\AxiomC{$X \lcut A,B, Y \rcut $}
\RightLabel{$(\lor)$}
\UnaryInfC{$X \lcut A\lor B, Y \rcut$}
\DisplayProof

\\[1.5em]

\AxiomC{$X[\bBox A, \bullet \{ A \}]$}
\RightLabel{$(\bBox)$}
\UnaryInfC{$X[\blacksquare A]$}
\DisplayProof

&

\AxiomC{$X[\bullet \{ Y, A \}, \blacklozenge A]$}
\RightLabel{$(\bdiam_{1})$}
\UnaryInfC{$X[\bullet \{ Y \}, \blacklozenge A]$}
\DisplayProof

&

\AxiomC{$X[\circ \{ Y, \blacklozenge A \}, A]$}
\RightLabel{$(\bdiam_{2})$}
\UnaryInfC{$X[\circ \{ Y, \blacklozenge A \}]$}
\DisplayProof

\\[1.5 em]

\AxiomC{$X[\Box A, \circ \{ A \}]$}
\RightLabel{$(\Box)$}
\UnaryInfC{$X[\Box A]$}
\DisplayProof

&

\AxiomC{$X[\circ \{ Y, A \}, \Diamond A]$}
\RightLabel{$(\diam_{1})$}
\UnaryInfC{$X[\circ \{ Y \}, \Diamond A]$}
\DisplayProof

&

\AxiomC{$X[\bullet\{ Y, \Diamond A \}, A]$}
\RightLabel{$(\diam_{2})$}
\UnaryInfC{$X[\bullet\{ Y, \Diamond A \}]$}
\DisplayProof

\end{tabular}

\end{center}

\end{definition}

We now aim to define propagation rules for deep nested calculi. To do this, we follow the work in \cite{GorPosTiu11} and first introduce path axiom inverses, compositions of path axioms, and the completion of a set of path axioms in order to define the corresponding set of equivalent propagation rules. Additions of these propagation rules to $\DKT$ will yield cut-free, sound, and complete deep nested calculi for logics $\Kt + \Pset$. Note that we define $\langle ? \rangle^{-1} = \Diamond$ if $\ques = \bldia$, and $\langle ? \rangle^{-1} = \bldia$, if $\ques = \Diamond$.

\begin{definition}[Path Axiom Inverse~\cite{GorPosTiu11}] If $F$ is a path axiom of the form $\dia{F_1} \cdots \dia{F_n} A \rightarrow \dia{F} A$, then we define the \emph{inverse of $F$} to be

\begin{center}
$I(F) := \dia{F_n}^{-1} \cdots \dia{F_1}^{-1} A \rightarrow \dia{F}^{-1} A$
\end{center}

\noindent
Given a set of path axioms $P$, we define the \emph{set of inverses} to be the set $I(P) := \{I(F) \ | \ F \in P\}$.

\end{definition}

\begin{definition}[Composition of Path Axioms~\cite{GorPosTiu11}] Given two path axioms

\begin{center}
$F := \dia{F_1} \cdots \dia{F_n} A \rightarrow \dia{F} A$ and $G := \dia{G_1} \cdots \dia{G_m} A \rightarrow \dia{G} A$
\end{center}

\noindent
we say  \emph{$F$ is composable with $G$ at $i$} iff $\dia{F} = \dia{G_{i}}$. We define the \emph{composition}

\begin{center}
$F \triangleright^{i} G := \dia{G_1} \cdots \dia{G_{i-1}} \dia{F_1} \cdots \dia{F_n} \dia{G_{i+1}} \cdots \dia{G_m} A \rightarrow \dia{G} A$
\end{center}

\noindent
when $F$ is composable with $G$ at $i$.

Using these individual compositions, we define the following \emph{set of compositions}:

\begin{center}
$F \triangleright G := \{F \triangleright^{i} G \ | \ \text{$F$ is composable with $G$ at $i$} \} $
\end{center}

\end{definition}

\begin{example} As an example, we can compose the axiom $\Diamond \Diamond A \rightarrow \blacklozenge A$ with $\blacklozenge \Diamond A \rightarrow \Diamond A$ to obtain $\Diamond \Diamond \Diamond A \rightarrow \Diamond A$.

\end{example}

\begin{definition}[Completion~\cite{GorPosTiu11}] The \emph{completion} of a set $P$ of path axioms, written $P^{*}$, is the smallest set of path axioms containing $P$ such that\\

(1) $\Diamond A \rightarrow \Diamond A, \bldia A \rightarrow \bldia A \in P^{*}$\\

(2) If $F,G \in P^{*}$ and $F$ is composable with $G$, then $F \triangleright G \subseteq P^{*}$.

\end{definition}

After introducing further notions necessary to define the propagation rules, we will give examples showing the significance of defining the rules relative to the \emph{completion} of a set of path axioms, rather than defining the rules relative to just the given set of path axioms. As will be shown, without defining the rules relative to the completion, the corresponding set of rules would not be enough to achieve completeness of the resulting calculus.

Let us now recall the notion of a propagation graph and the notion of a path in a propagation graph from~\cite{GorPosTiu11}. We introduce these concepts using the diamond rules of $\DKT$ as an example. The diamond rules ($\Diamond_{1}$), ($\Diamond_{2}$), ($\bldia_{1}$), ($\bldia_{2}$) can be read bottom-up as propagating formulae to nodes in the tree of a sequent.

For example, the ($\Diamond_{1}$) rule propagates an $A$ to a node along a $\circ$-edge, whereas the ($\Diamond_{2}$) rule propagates an $A$ backward along a $\bullet$-edge. Similarly, the ($\bldia_{1}$) rule propagates an $A$ forward to a node along a $\bullet$-edge, and the ($\bldia_{2}$) rule propagates an $A$ backward along a $\circ$-edge. These movements are represented in the diagram below:

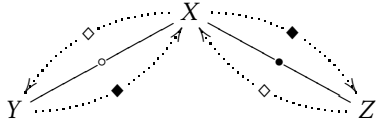
\begin{center}
\begin{tabular}{c}
\xymatrix{
		& &   & X\ar@{-}[dll]|-{\circ}\ar@/^-1pc/@{.>}[dll]|-{\Diamond} \ar@{-}[drr]|-{\bullet}\ar@/^1pc/@{.>}[drr]|-{\blacklozenge} 	&   & 	&   		\\
		& Y\ar@/^-1pc/@{.>}[urr]|-{\blacklozenge}  &  &  	&	& Z\ar@/^1pc/@{.>}[ull]|-{\Diamond} &   
}
\end{tabular}
\end{center}


This understanding of how formulae are propagated is crucial to define the propagation rules for deep nested calculi. In fact, as will be explained below, each path axiom can be read as an instruction that expresses how to propagate a formula along some path. We therefore give a precise definition of the \emph{propagation graph} of a sequent, which explicitly specifies how formulae may move when being propagated throughout the tree of a sequent.

\begin{definition}[Propagation Graph~\cite{GorPosTiu11}] Let $X$ be a nested sequent where $N$ is the set of nodes in $tree(X)$. We define the \emph{propagation graph $PG(X) = ( N,E )$ of $X$} to be the directed graph with set of nodes $N$ and set of edges $E$ each labeled with either a $\Diamond$ or $\blacklozenge$ as follows:

\begin{enumerate}

\item For every node $n \in N$ and $\circ$-child $m$ of $n$, we have a labeled edge $(n,m,\Diamond) \in E$ 
and a labeled edge $(m,n,\blacklozenge) \in E$.

\item For every node $n \in N$ and $\bullet$-child $m$ of $n$, we have a labeled edge $(n,m,\blacklozenge) \in E$ 
and a labeled edge $(m,n,\Diamond) \in E$. 

\end{enumerate}
\end{definition}

\begin{lemma}\label{DisEquiv_Sequents_Identical_PG}
Suppose that $X$ and $Y$ are display equivalent nested sequents. Then, $PG(X) = PG(Y)$.
\end{lemma}

\begin{proof} We prove the result by induction on the minimum number of display inferences needed to derive $Y$ from $X$.

\textbf{Base case.} Assume w.l.o.g. that $X = Z,{\circ}\{W\}$ and $Y = {\bullet}\{Z\},W$ so that $Y$ is derivable from $X$ with a single application of a display rule. Let $PG(Z) = (N_{1},E_{1})$ and $PG(W) = (N_{2},E_{2})$ with $n_{1}$ the root of $tree(Z)$ and $n_{2}$ the root of $tree(W)$. Observe that $PG(X) = (N,E)$, where $N = N_{1} \cup N_{2}$ and $E = E_{1} \cup E_{2} \cup \{(n_{1},n_{2},\Diamond),(n_{2},n_{1},\bldia)\}$, 
which is identical to $PG(Y)$ by definition.

\textbf{Inductive step.} Suppose that $n+1$ is the minimum number of display inferences needed to derive $Y$ from $X$. It follows that there exists a nested sequent $Z$ such that $Z$ is derivable from $X$ with one display inference, and $Y$ is derivable from $Z$ with $n$ applications of the display rules. By the base case we know that $PG(X)= PG(Z)$, and by the inductive hypothesis, $PG(Z) = PG(Y)$.
\end{proof}

\begin{definition}[Path~\cite{GorPosTiu11}] A \emph{path} is a sequence of nodes and diamonds (labeling edges) of the form:

\begin{center}
$n_{1},\ques_{1},n_{2},\ques_{2},...,\ques_{k-1},n_{k}$
\end{center}

\noindent
in the propagation graph $PG(X)$ such that $n_{i}$ is connected to $n_{i+1}$ by an edge labeled with $\ques_{i}$. Note that we allow repetitions of nodes along a path (e.g. $n, \Diamond, m, \blacklozenge, n$ is a path). For a given path $\pi = n_{1},\ques_{1},n_{2},\ques_{2},...\ques_{k-1},n_{k}$, we define the \emph{string of $\pi$} to be the string of diamonds $\Pi = \ques_{1} \ques_{2}...\ques_{k-1}$.

\end{definition}

\begin{definition}[Deep Nested Propagation Rules~\cite{GorPosTiu11}] Let $\Pset$ be a set of path axioms. The set of propagation rules $\propdeep$ contains all rules of the form:

\begin{center}
\AxiomC{$X[\ques A]_{i}[A]_{j}$}
\RightLabel{$\dprop$}
\UnaryInfC{$X[\ques A]_{i}[\emptyset]_{j}$}
\DisplayProof
\end{center}

\noindent
where there is a path $\pi$ from $i$ to $j$ in the propagation graph of the premise and $\Pi A \rightarrow \ques A \in (P \cup I(P))^{*}$ with $\Pi$ the string of $\pi$.


\end{definition}

It should be noted that two different sets $P$ and $P'$ of path axioms can generate the same set of propagation rules, i.e. $(P \cup I(P))^{*} = (P' \cup I(P'))^{*}$. For example, both $\{A \rightarrow \Diamond A, \blacklozenge \Diamond A \rightarrow \Diamond A\}$ and $\{A \rightarrow \Diamond A, \blacklozenge A \rightarrow \Diamond A, \Diamond \Diamond A \rightarrow \Diamond A\}$ yield the same set of propagation rules, which would provide a deep nested calculus for tense $\mathsf{S5}$.

\begin{example}[Necessity of Inverses] Let us now demonstrate why inverses must be taken into account when defining propagation rules. Suppose that we did not define the set of propagation rules relative to the set $(\{\Diamond \Diamond A \rightarrow \Diamond A\} \cup \{\blacklozenge \blacklozenge A \rightarrow \bldia A\})^{*}$, but rather, we defined the set of propagation rules relative to the set $\{\Diamond \Diamond A \rightarrow \Diamond A\}^{*}$. All propagation rules in this restricted set are of the form below (where there is a path of the form $i$,$\Diamond, \ldots, \Diamond$,$j$ of length $n \geq 1$ from $i$ to $j$):


\begin{center}
\AxiomC{$X[\Diamond A]_{i}[A]_{j}$}
\RightLabel{$\dprop$}
\UnaryInfC{$X[\Diamond A]_{i}[\emptyset]_{j}$}
\DisplayProof
\end{center}

We now explain why this restricted set of propagation rules--that does not take inverses into account---leads to an incomplete calculus. Below, we attempt to give a root-first derivation of $I(\Diamond \Diamond p \rightarrow \Diamond p) = \blacklozenge \blacklozenge p \rightarrow \bldia p$, which is a theorem of the logic $\Kt + \Diamond \Diamond A \rightarrow \Diamond A$ and should therefore be derivable:

\begin{center}
\begin{tabular}{c}
\AxiomC{$\bullet \{ \bullet \{ \overline{p} \} \}, \bldia p$}
\RightLabel{$\bboxr$}
\UnaryInfC{$\bullet \{ \bBox \overline{p} \}, \bldia p$}
\RightLabel{$\bboxr$}
\UnaryInfC{$\bBox \bBox \overline{p}, \bldia p$}
\RightLabel{$\disr$}
\UnaryInfC{$\bBox \bBox \overline{p} \vee \bldia p$}
\RightLabel{=}
\dottedLine
\UnaryInfC{$\blacklozenge \blacklozenge p \rightarrow \bldia p$}
\DisplayProof
\end{tabular}
\end{center}

Observe that no propagation rule from the restricted set is applicable to the top sequent of the derivation because no propagation rule acts along a path of the form $i$, $\blacklozenge, \ldots, \blacklozenge$, $j$. However, if we allow ourselves to define the propagation rules relative to the set $(\{\Diamond \Diamond A \rightarrow \Diamond A\} \cup \{\blacklozenge \blacklozenge A \rightarrow \bldia A\})^{*}$, then we also have the following rules in our calculus (where there is a path of the form $i$,$\blacklozenge, \ldots, \blacklozenge$,$j$ of length $n \geq 1$ from $i$ to $j$):


\begin{center}
\AxiomC{$X[\bldia A]_{i}[A]_{j}$}
\RightLabel{$\dprop$}
\UnaryInfC{$X[\bldia A]_{i}[\emptyset]_{j}$}
\DisplayProof
\end{center}

Using this rule we can complete the derivation by deriving the top sequent of the above derivation from the initial sequent $\bullet \{ \bullet \{ \overline{p}, p \} \}, \bldia p$:

\begin{center}
\AxiomC{}
\RightLabel{$\id$}
\UnaryInfC{$\bullet \{ \bullet \{ \overline{p}, p \} \}, \bldia p$}
\RightLabel{$\dprop$}
\UnaryInfC{$\bullet \{ \bullet \{ \overline{p} \} \}, \bldia p$}
\DisplayProof
\end{center}

\end{example}

\begin{example}[Necessity of Compositions] Suppose we are given the set $P = \{ \Diamond \blacklozenge \Diamond A \rightarrow \Diamond A, \Diamond \Diamond A \rightarrow \blacklozenge A\}$. One of the composition formulae derivable in the logic $\Kt + P$ is $\Diamond \Diamond \Diamond \Diamond A \rightarrow \Diamond A$. Our example below demonstrates the necessity of defining $\propdeep$ relative to the completion $(P \cup I(P))^{*}$ (which takes into account compositions) instead of just $P$.

If we define our propagation rules relative to just $P$, then we will have the following two propagation rules in our calculus:

\begin{center}
\begin{tabular}{c @{\hskip 3em} c}
\AxiomC{$X[\Diamond A]_{i}[A]_{j}$}
\RightLabel{$\dprop$}
\UnaryInfC{$X[\Diamond A]_{i}[\emptyset]_{j}$}
\DisplayProof

&

\AxiomC{$X[\bldia A]_{k}[A]_{n}$}
\RightLabel{$\dprop$}
\UnaryInfC{$X[\bldia A]_{k}[\emptyset]_{n}$}
\DisplayProof
\end{tabular}
\end{center}

The left rule is applicable when there is a path of the form $i$, $\Diamond$, $n_{1}$, $\bldia$, $n_{2}$, $\Diamond$, $j$ from node $i$ to $j$, and the right rule is applicable when there is a path of the form $k$, $\Diamond$, $n_{1}$, $\Diamond$, $n$ from $k$ to $n$ in the respective propagation graphs.

We now attempt to derive $\Diamond \Diamond \Diamond \Diamond p \rightarrow \Diamond p$, and show that no sequence of rules applied backward can give a proof of the formula:

\begin{center}
\AxiomC{$\circ \{ \circ \{ \circ \{ \circ \{  \overline{p} \} \} \} \}, \Diamond p$}
\RightLabel{$\boxr \times 4$}
\UnaryInfC{$\Box \Box \Box \Box \overline{p}, \Diamond p$}
\RightLabel{$\disr$}
\UnaryInfC{$\Box \Box \Box \Box \overline{p} \vee \Diamond p$}
\dottedLine
\RightLabel{=}
\UnaryInfC{$\Diamond \Diamond \Diamond \Diamond p \rightarrow \Diamond p$}
\DisplayProof
\end{center}

None of the rules in $\DKT$ or in the restricted set of propagation rules are bottom-up applicable to the top sequent. However, since $\Diamond \Diamond \Diamond \Diamond A \rightarrow \Diamond A \in (P \cup I(P))^{*}$, if we allow the addition of propagation rules to correspond to axioms in $(P \cup I(P))^{*}$ rather than just $P$, then we have the following rule in our calculus (where there is a path of the form $c$, $\Diamond$, $n_{1}$, $\Diamond$, $n_{2}$, $\Diamond$, $n_{3}$, $\Diamond$, $p$ from $c$ to $p$):

\begin{center}
\AxiomC{$X[\Diamond A]_{c}[A]_{p}$}
\RightLabel{$\dprop$}
\UnaryInfC{$X[\Diamond A]_{c}[\emptyset]_{p}$}
\DisplayProof
\end{center}

This can be used to prove the formula $\Diamond \Diamond \Diamond \Diamond p \rightarrow \Diamond p$ by deriving the top sequent in the above derivation from the initial sequent $\circ \{ \circ \{ \circ \{ \circ \{  \overline{p}, p \} \} \} \}, \Diamond p$:

\begin{center}
\AxiomC{}
\RightLabel{$\id$}
\UnaryInfC{$\circ \{ \circ \{ \circ \{ \circ \{  \overline{p}, p \} \} \} \}, \Diamond p$}
\RightLabel{$\dprop$}
\UnaryInfC{$\circ \{ \circ \{ \circ \{ \circ \{  \overline{p} \} \} \} \}, \Diamond p$}
\DisplayProof
\end{center}

\end{example}

\begin{lemma}[\cite{GorPosTiu11}]\label{dktstrucadmiss} The following rules are admissible in $\DKT + \propdeep$:

\begin{center}

\begin{tabular}{c @{\hskip 1em} c @{\hskip 1em} c @{\hskip 1em} c}

\AxiomC{$X[Y]$}
\RightLabel{$\wk$}
\UnaryInfC{$X[Y, Z]$}
\DisplayProof

&

\AxiomC{$X[Y, Y]$}
\RightLabel{$\ctr$}
\UnaryInfC{$X[Y]$}
\DisplayProof

&

\AxiomC{$X, \circw Y \}$}
\RightLabel{$\rf$}
\UnaryInfC{$\circb X \}, Y $}
\DisplayProof

&

\AxiomC{$X, \circb Y \}$}
\RightLabel{$\rp$}
\UnaryInfC{$\circw X \}, Y $}
\DisplayProof

\end{tabular}

\end{center}

\end{lemma}

\begin{lemma}[\cite{GorPosTiu11}]\label{sktdktequiv} Let $P$ be a set of path axioms. Every derivation in $\SKT + \Pd$ of a sequent $X$ is transformable [effectively relatable] to a derivation in $\DKT + \propdeep$, and vice-versa.
\end{lemma}
We have inserted the term ``effectively relatable'' to emphasize that the proof in~\cite{GorPosTiu11} is via a local transformation (i.e. a step-by-step procedure defined at the level of the proof rules in the derivation). Indeed, the forward direction of the above lemma is shown by induction on the height of the given derivation (\cite[Lem.~6.13]{GorPosTiu11}), and implies that we can effectively transform a derivation $\der$ in $\DKT + \propdeep$ into a derivation $\der'$ of the same end sequent in $\SKT + \Pd$. The reverse direction follows from the fact that $\SKT + \Pd$ can mimic propagation rules (\cite[Lem.~6.12]{GorPosTiu11}). Also, observe that the above lemma implies cut-free completeness for each deep nested calculus $\DKT + \propdeep$.

\begin{theorem}[\cite{GorPosTiu11}] Let $P$ be a set of path axioms. $\Kt+\Pset \vdash A$ iff $A$ is cut-free derivable in $\DKT+\propdeep$.

\end{theorem}


\section{Nested Sequents and Labeled Polytrees}

    In this section we show how to translate (back and forth) a nested sequent into a labeled polytree 
(called a \emph{labeled UT} in \cite{CiaLyoRam18}). These graphical structures facilitate the translations between nested and labeled proofs. 

We write~$V=V_{1}\sqcup V_{2}$ to mean that~$V=V_{1}\cup V_{2}$ and $V_{1}\cap V_{2}=\emptyset$. The multiset union of multisets~$M_{1}$ and~$M_{2}$ is denoted~$M_{1}\uplus M_{2}$. A \textit{labeling function~$L$} is a map from a set~$V$ to a multiset of tense formulae. For labeling functions~$L_{1}$ and~$L_{2}$ on the sets~$V_{1}$ and~$V_{2}$ respectively, let~$L_{1}\cup L_{2}$ be the labeling function on~$V_{1}\cup V_{2}$ defined as follows:
\[
(L_{1}\cup L_{2})(x)=\begin{cases}
L_{1}(x)			&	x\in V_{1}, x\not\in V_{2}	\\
L_{2}(x)			&	x\not\in V_{1}, x\in V_{2}	\\
L_{1}(x) \uplus L_{2}(x)	&	x\in V_{1}, x\in V_{2}
\end{cases}
\]
A \textit{labeled graph}~$(V,E,L)$ is a directed graph~$(V,E)$ equipped with a labeling function~$L$ on~$V$.

\begin{definition}[Labeled Graph Isomorphism]
We say that two labeled graphs $G_{1}$ = $(V_{1},E_{1},L_{1})$ and~$G_{2}=(V_{2},E_{2},L_{2})$ are isomorphic (written $G_{1} \cong G_{2}$) if and only if there is a function~$f:V_{1}\to V_{2}$ such that:

\begin{flushleft}
(i) $f$ is bijective;\\
(ii)~for every $x,y\in V_{1}$, $(x,y)\in E_{1}$ iff $(fx,fy)\in E_{2}$;\\ (iii)~for every~$x\in V_{1}$, $L_{1}(x)=L_{2}(fx)$.
\end{flushleft}
\end{definition}

\begin{definition}[Labeled Polytree]\label{Labeled_Polytree_Def}
A \emph{labeled polytree} is a labeled graph whose underlying (i.e. undirected) graph is a tree, \textit{i.e.} there exists exactly one path of undirected edges between every pair of distinct nodes.
\end{definition}


\begin{example} The following two graphs represent labeled polytrees, where each node is decorated with a multiset $M_{i}$ of formulae:

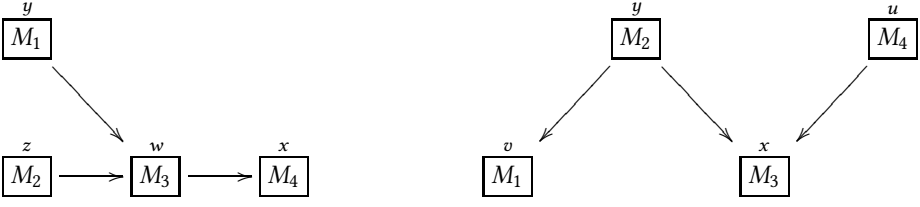
\begin{center}
\begin{tabular}{c @{\hskip 6em} c}
\xymatrix{
			   \overset{y}{\boxed{M_{1}}}\ar[dr]    &		    &	  	\\
			 \overset{z}{\boxed{M_{2}}} \ar[r]      &	\overset{w}{\boxed{M_{3}}}  \ar[r]	    &	 \overset{x}{\boxed{M_{4}}}  	    
}

&

\xymatrix{
		 &  \overset{y}{\boxed{M_{2}}}\ar[dr]\ar[dl]	& 	&  \overset{u}{\boxed{M_{4}}}\ar[dl]		\\
		  \overset{v}{\boxed{M_{1}}} &	&  \overset{x}{\boxed{M_{3}}}	& 		 \\
}
\end{tabular}
\end{center}
\end{example}

Polytrees have been discussed in the graph theory literature and have also found applications in computer science~\cite{HueCam93,RebPer13}.

\subsection{Interpreting a Nested Sequent as a Labeled Polytree}\label{sec:interpret-nested-as-labeled}

Every nested sequent has a natural interpretation as a labeled tree with two types of directed edges: $\overset{\circ}{\rightarrow}$ and $\overset{\bcirc}{\rightarrow}$ \cite{Kas94,GorPosTiu11}. If we interpret every directed edge $\alpha\overset{\bcirc}{\rightarrow}\beta$ as the directed edge $\alpha\overset{\circ}{\leftarrow}\beta$, we can then interpret every nested sequent as a connected labeled graph with a \textit{single} type of directed edge (so we can drop the~$\circ$ symbol altogether). Moreover, it is easy to see that its underlying graph (\textit{i.e.} the undirected graph obtained by treating all edges as undirected) is a tree, and that every nested sequent can be interpreted naturally as a labeled polytree.


\begin{example}[Transforming a Nested Sequent into a Labeled Sequent] First, interpret the nested sequent $A, \circw B, \circb \varepsilon \} \}, \circb D, E, \circb F \}, \circw G \} \}$ as the labeled tree with two types of directed edges, below left. Next, convert the labeled tree to a labeled polytree (with directed edges of a single type) by reading each $\alpha\overset{\bcirc}{\rightarrow}\beta$ as $\alpha\leftarrow\beta$ (below right) and remove the $\circ$-typing from the remaining edges.

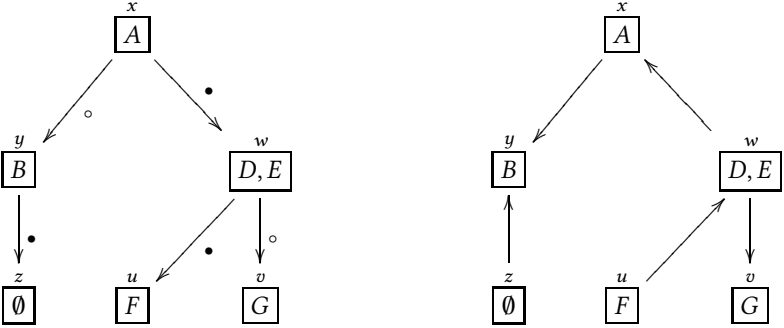
\begin{center}
\begin{tabular}{@{\hskip -1.5em} c@{\hspace{1em}}c}
\xymatrix{
        &    &    \overset{x}{\boxed{A}} \ar[dl]^{\circ} \ar[dr]^{\bcirc}    &    &    \\
        &    \overset{y}{\boxed{B}}\ar[d]^{\bcirc}    &    &\overset{w}{\boxed{D,E}}\ar[dl]^{\bcirc}\ar[d]^{\circ}    &     \\
        & \overset{z}{\boxed{\emptyset}}   &  \overset{u}{\boxed{F}}  &    \overset{v}{\boxed{G}} & 
}

&

\xymatrix{
        &    &    \overset{x}{\boxed{A}} \ar[dl]    &    &    \\
        &    \overset{y}{\boxed{B}}    &    & \overset{w}{\boxed{D,E}} \ar[ul]\ar[d]    &     \\
        & \overset{z}{\boxed{\emptyset}}\ar[u]   &  \overset{u}{\boxed{F}}\ar[ur]  &  \overset{v}{\boxed{G}}  &   
}
\end{tabular}
\end{center}

\end{example}
For concreteness let us formally define the map~$\dUTG$ from a nested sequent to a labeled polytree. 

\begin{definition}[The Translation $\dUTG$] We define the map~$\dUTG_{x}(X)$ recursively on the structure of the input nested sequent~$X$ as follows:\footnote{Thank you to the anonymous reviewer who suggested this definition.}
\begin{eqnarray*}
\dUTG_{x}(\varepsilon) &:=& (\emptyset, \emptyset, \emptyset)\\
\dUTG_{x}(A) &:=& (\{x\},\emptyset,\{(x,A)\} )\\
\dUTG_{x}(X_1,X_2) &:=&  (V_1 \cup V_2, E_1 \cup E_2, L_1 \cup L_2) \text{ where } \dUTG_{x}(X_{i}) = (V_i,E_i,L_i)\\
\dUTG_{x}(\circ\{X\}) &:=& (V \cup \{x\}, E \cup \{(x,y)\}, L) \text{ where } \dUTG_{y}(X) = (V,E,L) \text{ and $y$ is fresh. }\\
\dUTG_{x}(\bullet\{X\}) &:=& (V \cup \{x\}, E \cup \{(y,x)\}, L) \text{ where } \dUTG_{y}(X) = (V,E,L) \text{ and $y$ is fresh. }
\end{eqnarray*}
\end{definition}


\begin{example} The labeled polytree $\dUTG_{x}(X)$ := $(V, E, L)$ corresponding to the nested sequent $X := A, \circw B, \circb C \} \}, \circb D \}$ is shown below:

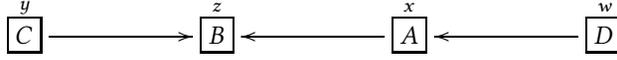
\begin{center}
\begin{tabular}{c} 
\xymatrix{
\overset{y}{\boxed{C}}\ar[rr]  &   &   \overset{z}{\boxed{B}} & &  \overset{x}{\boxed{A}}\ar[ll]  &  & \overset{w}{\boxed{D}}\ar[ll]    
}
\end{tabular}
\end{center}
\end{example}

\begin{definition}[Labeled Polytree Merge and Subgraph] Let~$G \oplus_{x} H$ denote the labeled polytree obtained as the graph union of labeled polytrees~$G$ and~$H$ that have a single vertex~$x$ in common, such that the label of~$x$ in~$G \oplus_{x} H$ (i.e. the multiset of tense formulae that~$x$ maps to under the labeling function of~$G \oplus_{x} H$) is the union of the labels of the vertex~$x$ in~$G$ and in~$H$. We refer to $G \oplus_{x} H$ as the \emph{merge} of two polytrees. 

We say that a $H$ is a \emph{labeled polytree subgraph} of a labeled polytree $G$ if and only if there exists a labeled polytree $H'$ and a vertex $x$ in $G$ such that $G = H' \oplus_{x} H$. We use $G \lcut H \rcut_{x}$ both as a name for the labeled polytree $G$ and to denote that $H$ is a labeled polytree subgraph of $G$.
\end{definition}

\begin{example} The labeled polytree $G \lcut H \rcut_{x} = H' \oplus_{x} H$, where $x$ is the common vertex between $H'$ and~$H$, is shown below left. The top labeled polytree below right is~$H'$ and the other is $H$.
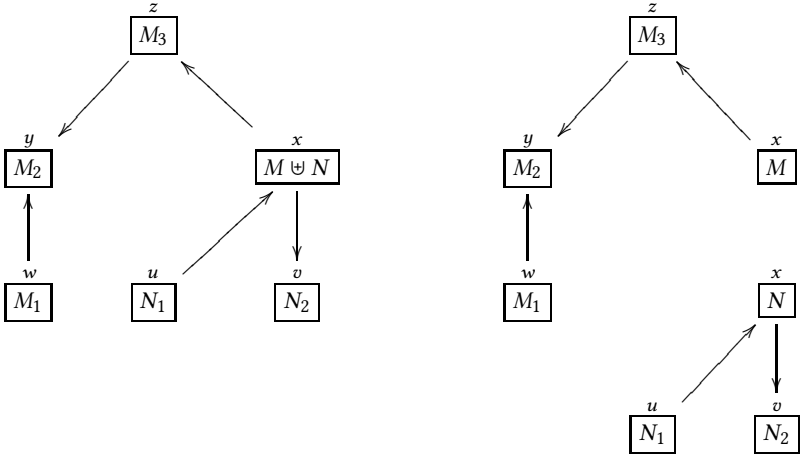
\begin{center}
\begin{small}
\begin{tabular}{c@{\hspace{2cm}}c}
\xymatrix{
							& \overset{z}{\boxed{M_{3}}} \ar[dl]	&											\\
\overset{y}{\boxed{M_{2}}}		&								& \overset{x}{\boxed{M\uplus N}} \ar[ul]\ar[d]			 \\
\overset{w}{\boxed{M_{1}}}\ar[u]   &  \overset{u}{\boxed{N_{1}}}\ar[ur]  		&  \overset{v}{\boxed{N_{2}}}
}

&
\xymatrix{
							& \overset{z}{\boxed{M_{3}}} \ar[dl]	&											\\
\overset{y}{\boxed{M_{2}}}		&								& \overset{x}{\boxed{M}} \ar[ul]				\\
\overset{w}{\boxed{M_{1}}}\ar[u]  &   								& \overset{x}{\boxed{N}}\ar[d]   						\\
							& \overset{u}{\boxed{N_{1}}}\ar[ur]		& \overset{v}{\boxed{N_{2}}}			
}
\end{tabular}
\end{small}
\end{center}
\end{example}

For any labeled polytree~$(V,E,L)$ with $x \in V$ there exist partitions $V=V_{1}\sqcup \{x\}\sqcup V_{2}$, $E=E_{1}\sqcup E_{2}$, and $L=L_{1}\cup L_{2}$ such that~$G \lcut H \rcut_{x} = H' \oplus_{x} H =(V,E,L)$ with~$H' = (V_{1}\sqcup\{x\},E_{1},L_{1})$ and~$H = (V_{2}\sqcup\{x\},E_{2},L_{2})$. Clearly, $L(x)=L_{1}(x)\uplus L_{2}(x)$, and $H'$ and~$H$ are labeled polytrees. In other words, we view $H$ in $G \lcut H \rcut_{x} = H' \oplus_{x} H$ as the redex and $H'$ as the context.

Since nested sequents may be interpreted as trees with two types of edges ($\circ$-edges and $\bcirc$-edges), they possess a root node, whereas labeled polytrees do not possess a root in general. Nevertheless, the underlying tree structure of a labeled polytree permits us to view any node as the root, and the lemma below ensures that we obtain isomorphic labeled polytrees via the display rules regardless of the node where we begin the translation.

Note that the label $x$ in~$\dUTG_{x}$ simply denotes the name of the starting vertex of the translation so~$\dUTG_{x}(X)\cong \dUTG_{y}(X)$ for all labels~$x$ and $y$, and all nested sequents~$X$.





\begin{lemma}\label{isoresid}
For every label~$x$, and any nested sequents $X$ and~$Y$: $\dUTG_{x}(X,\circ \{Y\})\cong \dUTG_{x}(\bcirc \{X\},Y)$.
\end{lemma}
\begin{proof}
Observe that~$\dUTG_{x}(X,\circ \{Y\})$ is isomorphic to the labeled polytree obtained from the disjoint union of~$\dUTG_{x}(X)$ and~$\dUTG_{y}(Y)$ by the addition of an edge~$(x,y)$.
Meanwhile~$\dUTG_{x}(\bcirc \{X\},Y)$ is isomorphic to the labeled polytree obtained from the disjoint union of~$\dUTG_{y}(X)$ and~$\dUTG_{x}(Y)$ by the addition of an edge~$(y,x)$.
The result follows because~$\dUTG_{x}(X)\cong \dUTG_{y}(X)$ and~$\dUTG_{y}(Y)\cong \dUTG_{x}(Y)$. 
\end{proof}

%
%

\begin{corollary}\label{isoresidcor} For all labels $x$ and $y$, and nested sequents $X$ and $Y$, if $X$ and $Y$ are display equivalent, then $\dUTG_{x}(X)\cong \dUTG_{y}(Y)$.
\end{corollary}

\begin{proof} By repeated application of Lemma \ref{isoresid}.
\end{proof}

Henceforth we write~$\dUTG$ instead of~$\dUTG_{x}$ to reduce clutter when the name of the starting vertex is not important.

\subsection{Interpreting a Labeled Polytree as a Nested Sequent}
\label{sec:interpret-lab-polytree-as-nested}

Given a labeled polytree $G = ( V, E, L )$ we first pick a vertex $x \in V$ to compute the nested sequent $\UTGd_{x}(G)$. If $E = \emptyset$, then $\UTGd_{x}(G) = L(x)$ is the desired nested sequent. Otherwise, for all $n$~forward looking edges $(x,y_{i}) \in E$ (with $1 \leq i \leq n$) where $G[H_{i}]_{y_{i}}$, and for all $k$ backward looking edges $(z_{j},x) \in E$ (with $1 \leq j \leq k$) where $G[H_{j}']_{z_{j}}$, we define the image of~$\UTGd_{x}(G)$ as the nested sequent
\[
L(x),\circ\{\UTGd_{y_{1}}(H_{1})\},\ldots, \circ\{\UTGd_{y_{n}}(H_{n})\}, \bcirc\{\UTGd_{z_{1}}(H_{1}')\},\ldots, \bcirc\{\UTGd_{z_{k}}(H_{k}')\}
\]
Since the labeled polytrees $H_{1},\ldots,H_{n},H_{1}',\ldots,H_{k}'$ are smaller than~$G$, the recursive definition of~$\UTGd$ is well-founded.

\begin{lemma}\label{DE_For_Diff_Node}
For any labeled polytree $G = ( V, E, L )$, and for any vertices $x, y \in V$, the nested sequent $\UTGd_{x}(G)$ is derivable from $\UTGd_{y}(G)$ via the display rules $\rf$ and $\rp$.
\end{lemma}

\begin{proof}
We prove the result by induction on the length of the (unique) path~$dist(x,y)$ between $x$ and $y$. When $dist(x,y)=0$ we have $x=y$ and the claim holds.

\textbf{Base case.} Suppose that $dist(x,y) = 1$. There are two cases to consider: either there is a forward edge from $x$ to $y$, or there is a backward edge from $x$ to $y$. Without loss of generality, we consider only the first case. It follows that if there is a forward edge connecting $x$ to $y$, then since $\UTGd_{x}(G)$ is of the form $X, \circ \{ Y \}$, then $\UTGd_{y}(G) = \bullet \{ X  \}, Y$. It is easy to see that both sequents are display equivalent.

\textbf{Inductive step.} Suppose that $dist(x,y) = n+1$. Let $z$ represent the node one edge away from $x$ along the $n+1$ path to $y$. By the base case, $\UTGd_{x}(G)$ and $\UTGd_{z}(G)$ are display equivalent, and since the distance from $z$ to $y$ is $n$, we have that $\UTGd_{z}(G)$ is also display equivalent to $\UTGd_{y}(G)$ by the induction hypothesis. Hence, $\UTGd_{x}(G)$ is display equivalent to $\UTGd_{y}(G)$.
\end{proof}

When translating a labeled polytree we must choose a vertex as the starting point of our translation. The above lemma states that all nested sequents obtained from choosing a different vertex to translate from are mutually derivable from one another, i.e. they are derivable from each other by use of the display rules $\rp$ and $\rf$ only (hence, they are display equivalent). Due to this fact, we will omit the subscript when contextually permissible and simply write $\UTGd$ as the translation function. 

To clarify the translation procedure, we provide an example below of the various nested sequents obtained from translating at a different initial vertex.

\begin{example} Suppose we are given the labeled polytree $G = ( V, E, L )$ where $V = \{x,y,z\}$, $E = \{(x,y), (z,x)\}$, $L(x) = \{ A \}$, $L(y) = \{B, C\}$, and $L(z) = \{D\}$. A pictorial representation of the labeled polytree $G$ is given on the left with the corresponding nested sequent translations on the right:


\begin{center}
\begin{tabular}{c c}
\xymatrix{
\overset{y}{\boxed{B,C}} &	\overset{x}{\boxed{A}}\ar[l] &   \overset{z}{\boxed{D}}\ar[l] 
}

&

\begin{tabular}{c}
$\UTGd_{x}(G) = A, \circ \{ B, C \}, \bullet \{ D \}$  \\
$\UTGd_{y}(G) = B, C, \circb A, \circb D \} \}$ \\
$\UTGd_{z}(G) = D, \circw A, \circw B, C \} \}$
\end{tabular}

\end{tabular}
\end{center}

\end{example}

The following lemma ensures that the pieces $X$ and $Y$ of the nested sequent $\UTGd_{x}(G \lcut H \rcut_{x})  = \UTGd_{x}(H' \oplus_{x} H) = X, Y$ and the pieces $H$ and $H'$ of the labeled polytree $\dUTG_{x}(X, Y) = G \lcut H \rcut_{x} = H' \oplus_{x} H$ correctly map to each other under our translation functions.

\begin{lemma}\label{lem-contexts} (i) 
For every~$X$ and~$Y$, $\dUTG_{x}(X,Y)$ is the labeled polytree~$G \lcut H \rcut_{x} = H' \oplus_{x} H$ where $H'$ is the labeled polytree~$\dUTG_{x}(X)$ and $H$ is the labeled polytree~$\dUTG_{x}(Y)$.

(ii) For every labeled polytree~$G[H]_{x} = H' \oplus_{x} H$, $\UTGd_{x}(G \lcut H \rcut_{x})$ is a nested sequent of the form~$X,Y$ where~$X = \UTGd_{x}(H')$ and~$Y = \UTGd_{x}(H)$.
\end{lemma}
\begin{proof}
By construction of~$\dUTG$ and~$\UTGd$.
\end{proof}

Last, we note that due to the correspondence between nested sequents and labeled polytrees, the latter are easily translated to formulas in $\mathcal{L}_{\Kt}$ via the function $\UTGd$ and the function $\mathcal{I}$ from Section~\ref{Shallow_Nested_Intro_Subsection}.

\section{From Shallow Nested to Labeled Calculi}

We answer the following question: given a derivation $\der$ of~$A$ in~$\SKT+ \GPd$, is there a derivation $\der'$ of~$x:A$ in $\LKt+ \GPl$ that is \textit{effectively relatable to~$\der$}? 
 As a consequence, the structure of the output derivation is related to the structure of the input.
The constraint that the new derivation be \textit{effectively relatable} is crucial, for otherwise one could trivially relate~$\der'$ with the derivation~$\der$ as obtained from the following equivalences:
\[
\text{$\exists \der(\vdash_{\SKT+ \GPd}^{\der}A)$ \ \ \ iff \ \ \ $\Kt+\GP \vdash A$ \ \ \ iff \ \ \ $\exists \der'(\vdash^{\der'}_{\LKt+ \GPl}x:A)$}
\]
where the notation $\exists \der(\vdash_{\SKT+ \GPd}^{\der}A)$ and $\exists \der'(\vdash^{\der'}_{\LKt+ \GPl}x:A)$ is taken to mean that there exists a derivation $\der$ and a derivation $\der'$ such that $A$ and $x : A$ are derivable in $\SKT+ \GPd$ and $\LKt+ \GPl$, respectively. 

The point is that the derivations in~$\SKT+ \GPd$ and $\LKt+ \GPl$ obtained solely from the above equivalences would \textit{not} be defined via local transformations. Indeed, due to the existential operators, the structure of the two derivations would not be related in any meaningful way.

\subsection{Transforming a Labeled Graph~$G = (V,E,L)$ into a Labeled Sequent $\R, \Gamma$} Define~$\mathcal{R}=\{ Rxy \ | \ (x,y)\in E\}$ and
\[
\Gamma=\biguplus_{x\in V} x:L(x)
\]
where $x : L(x)$ represents the multiset $L(x)$ with each formula prepended with a label $x$.
\begin{example} The labeled graph $G = ( V, E, L )$ where $V = \{x, y, z \}$, $E = \{(x,y), (z,x) \}$, $L(x) = \{ A \}$, $L(y) = \{ B \}$, and $L(z) = \{ C \}$ corresponds to the labeled sequent $Rxy, Rzx, x:A, y:B, z:C$.

\end{example}

\subsection{Transforming a Labeled Sequent~$\mathcal{R},\Gamma$ into a Labeled Graph~$(V,E,L)$} Let~$V$ be the set of all labels occurring in $\R, \Gamma$. Define
\begin{align*}
E=\{(x,y) \ | \ Rxy\in\mathcal{R}\}	&&	L(x)=\{A \ | \ x:A \in \Gamma\}
\end{align*}
\begin{example} The labeled sequent $Rxy, Ryz, Rux, x:A, z:B, z:C, u:D$ becomes the labeled graph $G = ( V, E, L )$ where $V = \{x, y, z, u \}$, $E = \{(x,y), (y,z), (u,x) \}$, $L(x) = \{ A \}$, $L(y) = \emptyset$, $L(z) = \{ B,C \}$ and $L(u) = \{ D \}$.

\end{example}

The reader will observe that the translations are obtained rather directly. This is because the main difference between a labeled graph and a labeled sequent is notation. Therefore, for a given nested sequent $X$, we let $\dUTG(X)$ also represent the labeled sequent obtained from the labeled polytree of $X$. We follow this convention for the remainder of the paper and let $\dUTG(X)$ represent a labeled sequent.


Combining the previous results we obtain:
\begin{theorem}\label{DisToLab} Let $\GP$ be a set of general path axioms. If~$\der$ is a derivation of $X$ in~$\SKT + \GPd$, then there is an effective translation of~$\der$ to a derivation~$\der'$ of $\dUTG(X)$ in~$\LKt + \GPl$.

\end{theorem}
\begin{proof}
We prove the result by induction on the height of the given derivation.

\textbf{Base case.} The translation of an initial sequent $Y, p, \overline{p}$ in $\SKT+ \GPd$ gives the initial sequent $\dUTG_{x}(Y), x:p, x:\overline{p}$ in $\LKt + \GPl$, which proves the base case.

\textbf{Inductive step.} We show the inductive step for the rules $\disr$, $\bboxr$, ($\diam$), $\rp$, and $\dgp$. We note that weakening and contraction are admissible in $\der'$ due to Lemma~\ref{g3ktstrucadmiss}.

\begin{center}
\begin{tabular}{c @{\hskip 2em} c}
\AxiomC{$Y, A,B $}
\RightLabel{$(\lor)$}
\UnaryInfC{$Y, A\lor B$}
\DisplayProof

&

\AxiomC{$\dUTG_{x}(Y), x:A, x:B$}
\RightLabel{$\disr$}
\UnaryInfC{$\dUTG_{x}(Y), x:A \vee B$}
\DisplayProof
\end{tabular}
\end{center}

\begin{center}
\begin{tabular}{c @{\hskip 2em} c}
\AxiomC{$Y, \bullet \{ A \}$}
\RightLabel{$(\bBox)$}
\UnaryInfC{$Y, \blacksquare A$}
\DisplayProof

&

\AxiomC{$\dUTG_{x}(Y), Ryx, y: A$}
\RightLabel{$\bboxr$}
\UnaryInfC{$\dUTG_{x}(Y), x: \bBox A$}
\DisplayProof
\end{tabular}
\end{center}

\begin{center}
\begin{tabular}{c @{\hskip 2em} c}
\AxiomC{$Y, \circ \{ Z, A \}, \Diamond A$}
\RightLabel{$(\diam)$}
\UnaryInfC{$Y, \circ \{ Z \}, \Diamond A$}
\DisplayProof

&

\AxiomC{$\dUTG_{x}(Y), \dUTG_{y}(Z), Rxy, x:\Diamond A, y:A$}
\RightLabel{($\diam$)}
\UnaryInfC{$\dUTG_{x}(Y), \dUTG_{y}(Z), Rxy, x:\Diamond A$}
\DisplayProof
\end{tabular}
\end{center}		 

\begin{center}
\begin{tabular}{c @{\hskip 2em} c}
\AxiomC{$Y, \bullet \{ Z \}$}
\RightLabel{$\rp$}
\UnaryInfC{$\circ \{ Y \}, Z$}
\DisplayProof

&

\AxiomC{$\dUTG_{y}(Y, \bullet \{ Z \})$}
\RightLabel{Lem.~\ref{isoresid}}
\dottedLine\UnaryInfC{$\dUTG_{z}(\circ \{ Y \}, Z)$}
\DisplayProof
\end{tabular}
\end{center}

\begin{center}
\begin{tabular}{c@{\hskip 2em}c}
\AxiomC{$Y, \star_{n+1} \{ ... \star_{n+m} \{ Z \} ... \}$}
\RightLabel{$\dgp$}
\UnaryInfC{$Y, \star_{1} \{ ... \star_{n} \{ Z \} ... \}$}
\DisplayProof

&

\AxiomC{$\dUTG_{x}(Y), \R_{\Sigma}xy, \dUTG_{y}(Z)$}
\RightLabel{Lem.~\ref{g3ktstrucadmiss}}
\dashedLine
\UnaryInfC{$\dUTG_{x}(Y), \R_{\Pi}xy, \R_{\Sigma}xy, \dUTG_{y}(Z)$}
\RightLabel{$\lgp$}
\UnaryInfC{$\dUTG_{x}(Y), \R_{\Pi}xy, \dUTG_{y}(Z)$}
\DisplayProof
\end{tabular}
\end{center}



Because $\dUTG_{y}(Y, \bullet \{ Z \})$ and $\dUTG_{z}(\circ \{ Y \}, Z)$ are isomorphic, the premise and conclusion of~$\rp$ can be mapped to the same labeled sequent (thus, the two will be identical), and hence no rule is used for translating~$\rp$. In the above, this is denoted by the dotted line.
\end{proof}

\begin{example}\label{example8} We translate a derivation of $\bdiam \Diamond p \rightarrow \diam\bdiam p$ in $\SKT + NestSt(\{\bdiam \Diamond p \rightarrow \diam\bdiam p\})$ to a derivation in $\LKt + LabSt(\{\bdiam \Diamond p \rightarrow \diam\bdiam p\})$. 
\begin{center}
\begin{tabular}{c@{\hskip 1cm}c}
\AxiomC{$\bdiam p, \bcirc\{ p,\bar{p}\}, \bcirc\{ \diam\bdiam p\}$}
\RightLabel{(${\bdiam}$)}
\UnaryInfC{$\bdiam p, \bcirc\{ \bar{p}\}, \bcirc\{ \diam\bdiam p\}$}
\RightLabel{$\rp$}
\UnaryInfC{$\circ\{ \bdiam p, \bcirc\{ \bar{p}\} \}, \diam\bdiam p$}
\RightLabel{(${\diam}$)}
\UnaryInfC{$\circ\{ \bcirc\{ \bar{p}\} \}, \diam\bdiam p$}
\RightLabel{$\dgp$}
\UnaryInfC{$\bcirc\{ \circ\{ \bar{p}\} \}, \diam\bdiam p$}
\RightLabel{$\rp$}
\UnaryInfC{$\circ\{ \bar{p}\}, \circ\{\diam\bdiam p\}$}
\RightLabel{(${\Box}$)}
\UnaryInfC{$\Box \bar{p}, \circ\{\diam\bdiam p\}$}
\RightLabel{$\rf$}
\UnaryInfC{$\bcirc\{ \Box \bar{p} \}, \diam\bdiam p$}
\RightLabel{(${\bBox}$)}
\UnaryInfC{$\bBox \Box \bar{p}, \diam\bdiam p$}
\RightLabel{(${\lor}$)}
\UnaryInfC{$\bBox \Box \bar{p}\lor \diam\bdiam p$}
\RightLabel{=}
\dottedLine
\UnaryInfC{$\bdiam \Diamond p \rightarrow \diam\bdiam p$}
\DisplayProof
&
\AxiomC{$Rxu, Rzu, Ryx, Ryz, z:\overline{p}, x:\Diamond \bldia p, u: \bldia p, z : p$}
\RightLabel{(${\bdiam}$)}
\UnaryInfC{$Rxu, Rzu, Ryx, Ryz, z:\overline{p}, x:\Diamond \bldia p, u: \bldia p$}
\RightLabel{(${\diam}$)}
\UnaryInfC{$Rxu, Rzu, Ryx, Ryz, z:\overline{p}, x:\Diamond \bldia p$}
\RightLabel{${\lgp}$}
\UnaryInfC{$Ryx, Ryz, z:\overline{p}, x:\Diamond \bldia p$}
\RightLabel{(${\Box}$)}
\UnaryInfC{$Ryx, y: \Box \overline{p}, x:\Diamond \bldia p$}
\RightLabel{(${\bBox}$)}
\UnaryInfC{$x:\blacksquare \Box \overline{p}, x:\Diamond \bldia p$}
\RightLabel{(${\lor}$)}
\UnaryInfC{$x: \blacksquare \Box \overline{p} \vee \Diamond \bldia p$}
\RightLabel{=}
\dottedLine
\UnaryInfC{$\bdiam \Diamond p \rightarrow \diam\bdiam p$}
\DisplayProof
\end{tabular}
\end{center}

\end{example}

\begin{corollary}
\label{cor:modal-translation}
Let $M \subseteq \{\Pi A \rightarrow \Sigma A \ | \ \Pi, \Sigma \in \{\Diamond\}^{*}\}$ be a set of modal general path axioms. Every derivation in the shallow nested calculus $\SKT - \{(\bBox), (\blacklozenge)\} + NestSt(M)$ is translatable to a derivation in the labeled calculus $\LKt - \{(\bBox), (\blacklozenge)\} + LabSt(M)$.
\end{corollary}
The above corollary shows that our translations also hold for the modal (non-tense) fragments of the logics we consider. As detailed in the conclusion, this is useful since one can prove conservativity of the display calculi $\SKT - \{(\bBox), (\blacklozenge)\} + NestSt(M)$ over their modal fragments by translating derivations into $\LKt - \{(\bBox), (\blacklozenge)\} + LabSt(M)$ and invoking the soundness of the labeled calculus.

\section{From Labeled to Shallow Nested Calculi}

    In this section, we address the converse question: translating labeled proofs into shallow nested proofs, which will be achieved by translating through the deep nested calculi $\DKT + \propdeep$. In the base case for $\Kt$ when~$\GP=\emptyset$, i.e. for the calculus $\LKt$, it is fairly straightforward to effectively translate labeled derivations into nested derivations. As will be argued in Lemma~\ref{Lab-equiv}, every derivation in $\LKt$ which proves a labeled theorem of the form $x:A$, consists solely of labeled sequents which are translatable into nested notation. After providing the translation from $\LKt$ to $\SKT$, we explain a problem that arises when attempting to translate derivations from extensions of $\LKt$ to extensions of $\SKT$, and how we solve this problem for path extensions of $\Kt$.

The central issue complicating the reverse translation from labeled to nested for general path extensions of $\Kt$ is that structural rule extensions of $\LKt$ allow for labeled sequents to occur in derivations that are \emph{not} labeled polytree sequents. In other words, general path structural rules allow one to derive theorems with labeled sequents not in the domain of the translation function given in Section~\ref{sec:interpret-lab-polytree-as-nested}.\footnote{Note that the translation function $\UTGd$ is only defined for labeled polytree sequents.} This complication arises since our translation is only defined for labeled polytree sequents, and not for labeled sequents in general. Nevertheless, we can overcome this obstacle by considering labeled calculi for $\Kt$ extended with propagation rules for path axioms since every derivation can be transformed into one containing only (translatable) labeled sequents, i.e. labeled polytree sequents. In Section~\ref{Trans_Path_Ax_Ext_Subsection}, we explain this proof transformation procedure, followed by the translation from $\LKt + \Pl$ to $\SKT + \Pd$ that leverages $\DKT + \propdeep$ to facilitate the translation. Note that although the translation presented here takes a detour through a deep nested calculus, a direct translation from labeled to shallow nested could be provided; still, we opt for the latter approach since it allows us to exploit results from~\cite{GorPosTiu11} that simplify our work.

\subsection{Translating the Base Calculus}\label{sec:labeled-to-nested-Kt}

We first consider the converse translation for the base calculus $\LKt$.


\begin{definition}[Labeled Polytree Sequent] A \emph{labeled polytree sequent} is a labeled sequent whose graph is a labeled polytree.

\end{definition}



\begin{lemma}\label{Lab-equiv}
Every labeled derivation in~$\LKt$ of a labeled polytree sequent $\R,\Gamma$ consists solely of labeled polytree sequents.
\end{lemma}
\begin{proof} We argue by contradiction. Let~$\der$ be a derivation of~$\R,\Gamma$ in~$\LKt$ and suppose there is a labeled sequent~$\R',\Gamma'$ in~$\der$ that is not a labeled polytree sequent. By definition, the underlying undirected graph of the graph of~$\R',\Gamma'$ is not a tree. It follows that the graph of~$\R',\Gamma'$ is not connected or contains an undirected cycle.

If the graph of~$\mathcal{R}', \Gamma'$ is not connected, then by inspection of the rules of~$\LKt$, the graph of every sequent below $\mathcal{R}', \Gamma'$ in $\der$ is disconnected, including the graph of $\R,\Gamma$, which is a contradiction.
On the other hand, if it is connected, then the graph of~$\mathcal{R}', \Gamma'$ must contain an undirected cycle. Since the derivation ends with a labeled polytree sequent $\R,\Gamma$, it must be the case that every undirected cycle is deleted ultimately. The only rules that delete relational atoms in $\LKt$ are $\boxr$ and $\bboxr$. However, the eigenvariable conditions in these rules are not satisfied for labels occurring in an undirected cycle, so the undirected cycle cannot be eliminated, giving a contradiction.

Hence, every sequent occurring in a $\LKt$ derivation of a labeled polytree sequent is a labeled polytree sequent.
\end{proof}

The observation that $\LKt$ is complete relative to derivations consisting solely of labeled polytree sequents is useful for our translation work. Recognizing that this fact generalizes to $\LKt$ extended with propagation rules allows us to easily translate our labeled derivations into deep nested derivations, and then leverage Lemma~\ref{sktdktequiv} to complete the effective translation from labeled to shallow nested.

\begin{lemma}\label{G3KT_to_DKT} Every derivation in $\LKt$ consisting solely of labeled polytree sequents can be effectively translated to a derivation in $\DKT$.

\end{lemma}

\begin{proof} We prove this by induction on the height of the given derivation.

\textbf{Base case.} The translation of an initial sequent $\R, x:p, x:\overline{p}, \Gamma$ in $\LKt$ gives an initial sequent $\UTGd(\R, x:p, x:\overline{p}, \Gamma) = X[p, \overline{p}]$ in $\DKT$ which proves the base case.

\textbf{Inductive step.} We show the inductive step for the rules $\disr$, $\bboxr$, and ($\Diamond$); all remaining cases are similar. 

\begin{center}
\begin{tabular}{c @{\hskip 2em} c}
\AxiomC{$\R, \Gamma, x:A, x:B$}
\RightLabel{$\disr$}
\UnaryInfC{$\R, \Gamma, x:A \vee B$}
\DisplayProof

&

\AxiomC{$\UTGd(\R, \Gamma, x:A, x:B)$}
\RightLabel{=}
\dottedLine
\UnaryInfC{$X \lcut A,B \rcut $}
\RightLabel{$(\lor)$}
\UnaryInfC{$X \lcut A\lor B \rcut$}
\RightLabel{=}
\dottedLine
\UnaryInfC{$\UTGd(\R, \Gamma, x:A \vee B)$}
\DisplayProof
\end{tabular}
\end{center}

\begin{center}
\begin{tabular}{c @{\hskip 2em} c}
\AxiomC{$\R, Ryx, y: A, \Gamma$}
\RightLabel{$\bboxr$}
\UnaryInfC{$\R, x: \bBox A, \Gamma$}
\DisplayProof

&

\AxiomC{$\UTGd(\R, Ryx, y: A, \Gamma)$}
\RightLabel{=}
\dottedLine
\UnaryInfC{$X[\bullet \{ A \}]$}
\RightLabel{Lem.~\ref{dktstrucadmiss}}
\dashedLine
\UnaryInfC{$X[\bBox A, \bullet \{ A \}]$}
\RightLabel{$(\bBox)$}
\UnaryInfC{$X[\blacksquare A]$}
\RightLabel{=}
\dottedLine
\UnaryInfC{$\UTGd(\R, x: \bBox A, \Gamma)$}
\DisplayProof
\end{tabular}
\end{center}
We note that for the $\bboxr$ case above, we choose to translate from a label $z \neq y$ to ensure that the $\bboxr$ rule may be immediately applied. We are permitted to translate from whatever label we please due to Lemma~\ref{dktstrucadmiss}, which states the admissibility of the $\rf$ and $\rp$ rules. For example, if we consider the translation $\UTGd_{y}(\R, Ryx, y: A, \Gamma)$ starting at $y$, then the output nested sequent would be of the form $\circ\{X\}, A$. Applying the admissibility of $\rf$ yields a proof of the nested sequent $X, \bullet\{A\}$, which the $\bboxr$ rule may be applied to. By restricting our translation to start translating from a label $z \neq y$ however, we ensure that the $\bboxr$ may be immediately applied. Translating $\boxr$ inferences is performed in a similar fashion.

For the ($\Diamond$) case, there are two possible inferences in $\DKT$ depending on whether the unique undirected path from the node~$z$ (that we translate from in the premise of the last inference in the $\LKt$ derivation) to~$x$ encounters~$y$ (below rightmost) or not (below center). Note that below center $Y$ stands for all formulae in $\Gamma$ labeled with $y$, and below rightmost $X$ stands for all formulae from $\Gamma$ labeled with $x$.

\begin{center}
\begin{tabular}{c @{\hskip 1em} c @{\hskip 1em} c}
\AxiomC{$\R, Rxy, x:\Diamond A, y:A, \Gamma$}
\RightLabel{($\diam$)}
\UnaryInfC{$\R, Rxy, x:\Diamond A, \Gamma$}
\DisplayProof

&

\AxiomC{$\UTGd_{z}(\R, Rxy, x:\Diamond A, y:A, \Gamma)$}
\RightLabel{=}
\dottedLine
\UnaryInfC{$Z[\circ\{ Y, A \}, \Diamond A]$}
\RightLabel{$(\diam_{1})$}
\UnaryInfC{$Z[\circ\{ Y \}, \Diamond A]$}
\RightLabel{=}
\dottedLine
\UnaryInfC{$\UTGd_{z}(\R, Rxy, x:\Diamond A, \Gamma)$}
\DisplayProof

&

\AxiomC{$\UTGd_{z}(\R, Rxy, x:\Diamond A, y:A, \Gamma)$}
\RightLabel{=}
\dottedLine
\UnaryInfC{$Z[\bullet\{ X, \Diamond A \}, A]$}
\RightLabel{$(\diam_{2})$}
\UnaryInfC{$Z[\bullet\{ X, \Diamond A \}]$}
\RightLabel{=}
\dottedLine
\UnaryInfC{$\UTGd_{z}(\R, Rxy, x:\Diamond A, \Gamma)$}
\DisplayProof
\end{tabular}
\end{center}

Consider the directed labeled graph associated with the premise of~($\diam$). It contains a directed edge from~$x$ to~$y$. Since the undirected graph is a polytree by assumption, there is a unique undirected path from~$z$ to~$x$. If (i)~$y$ is not on this path, then the translation starting at~$z$ encounters~$x$ first. Since $x$ to~$y$ is a forward edge (i.e. it is in the same direction as the directed edge) it will be translated as~$\circ$; in this case apply the~$(\diam_{1})$ rule. Otherwise it must be that (ii)~$y$ is on the path. This means that the translation starting at~$z$ encounters~$y$ first. Since $y$ to~$x$ is a backward edge (i.e. its direction is the reverse of the directed edge) it will be translated as~$\bullet$; in this case apply the~$(\diam_{2})$ rule.
\end{proof}

\begin{theorem}
Every derivation in $\LKt$ of a labeled polytree sequent $\R,\Gamma$ is effectively relatable to a derivation of $\UTGd(\R,\Gamma)$ in $\SKT$.
\end{theorem}

\begin{proof} Let $\der$ be a derivation in $\LKt$ of a labeled polytree sequent $\R,\Gamma$. By Lemma~\ref{Lab-equiv}, $\der$ consists solely of labeled polytree sequents. Hence, by Lemma~\ref{G3KT_to_DKT} we can effectively transform $\der$ into a derivation $\der'$ in $\DKT$, and so, by Lemma~\ref{sktdktequiv} we can effectively transform $\der'$ into a derivation of $\UTGd(\R,\Gamma)$ in $\SKT$. The composition of these two effective transformations gives the desired effective transformation.
\end{proof}

The above argument does not always work for \textit{extensions of~$\LKt$} because additional structural rules may be capable of removing cycles in the following sense: the underlying undirected graph of the premise might have a cycle, yet the underlying undirected graph of the conclusion might not (this was not the case for any rule in~$\LKt$). For example, let us consider the rule ($\mathsf{Conf}$) for the confluence axiom $\blacklozenge \Diamond A \rightarrow \Diamond \blacklozenge A$: 
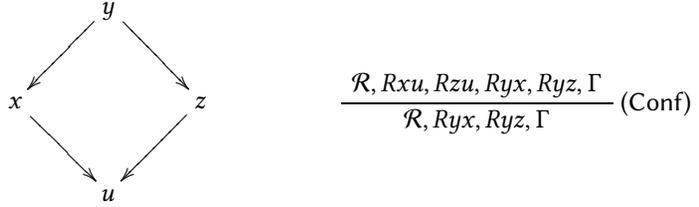
\begin{minipage}[t]{.45\textwidth}
\xymatrix{
 & y\ar[dl]\ar[dr] & \\
x\ar[dr] & & z\ar[dl] \\
 & u & 
}
\end{minipage}
\begin{minipage}[t]{.1\textwidth}
\
\end{minipage}
\begin{minipage}[t]{.45\textwidth}
\vspace{2.5em}
\begin{tabular}{c}
\AxiomC{$\mathcal{R},Rxu,Rzu,Ryx,Ryz,\Gamma$}
\RightLabel{($\mathsf{Conf}$)}
\UnaryInfC{$\mathcal{R},Ryx,Ryz,\Gamma$}
\DisplayProof
\end{tabular}
\end{minipage}

In a rule instance of ($\mathsf{Conf}$), the underlying undirected graph of the premise necessarily contains a cycle of the form shown above left. However, it need not be the case that the underlying undirected graph of the conclusion contains a cycle, since the edges from $x$ to $u$ and $z$ to $u$ (corresponding to the relational atoms $Rxu$ and $Rzu$) are deleted. As a consequence, a labeled derivation of a labeled formula~$x:A$ in $\LKt+\text{($\mathsf{Conf}$)}$ may contain labeled sequents that are not labeled polytree sequents. Therefore, such a derivation is not immediately translatable to a derivation in $\SKT + (\mathsf{Conf})$ via our methods because the derivation may contain sequents that are not in the domain of our translation.

For all general path extensions of $\Kt$, every shallow nested derivation can be translated into a labeled derivation; this fact implies that the space of shallow nested derivations corresponds to a subspace of the space of labeled derivations. Derivations of theorems in our labeled calculi may contain labeled sequents that are not labeled polytree sequents, showing that labeled derivations contain structures that go beyond those of the nested formalism. Nevertheless, we may invoke techniques present in~\cite{GorPosTiu11,LyoBer19} to pre-process each labeled derivation (in a labeled calculus for $\Kt$ extended with path axioms $\Pset$) in such a way that each is translatable to a shallow nested derivation, thus answering an open question in~\cite{CiaLyoRam18}.

\subsection{Translating the Path Axiom Extension}\label{Trans_Path_Ax_Ext_Subsection} 

We now show that the labeled calculus can be \emph{internalized} (also referred to as \emph{refinement} in~\cite{LyoBer19}) for $\Kt + \Pset$ (where $\Pset$ represents a set of path axioms), meaning that we can effectively transform any $\LKt + \Pl$ derivation of a labeled formula into one where every sequent is a labeled polytree sequent. 
 This internalization of proofs is interesting in its own right, and is also helpful in that the resulting labeled derivation is easily translatable into a derivation in $\DKT + \propdeep$. From there, we can invoke Lemma~\ref{sktdktequiv} to conclude the existence of an effective translation from $\LKt + \Pl$ derivations to $\SKT + \Pd$ derivations (since composing two effective procedures gives an effective procedure). 

The method of transforming every derivation in $\LKt + \Pl$ into a derivation consisting solely of labeled polytree sequents relies on the addition of propagation rules 
$\prop$ to the calculus (cf.~\cite{GorPosTiu11,LyoBer19,Sim94}). 
Such propagation rules simulate the $\lp$ rules, preserve disconnected and cyclic structures downward in a derivation, and, equivalently, preserve labeled polytree structure bottom-up in a derivation. 
The latter properties are significant because they allow us to make an argument similar to the one made in Lemma~\ref{Lab-equiv}, where we argue by contradiction that every labeled sequent occurring in a given derivation of a labeled formula $x :A$ must be a labeled polytree sequent.

The main technical lemma in this section is Lemma~\ref{Permute_Prop_and_Struc}, where we show that in the presence of propagation rules $\prop$, the structural rules $\Pl$ in $\LKt + \Pl$ can be eliminated from any proof. This allows for the effective transformation of any proof in an (unrestricted) labeled calculus $\LKt + \Pl$ into a proof in the associated internal labeled calculus $\LKt + \prop$ (Lemma~\ref{internal}). 
Proofs in the internal calculi $\LKt + \prop$ can then be effectively translated into derivations in $\DKT + \propdeep$. Once we prove these claims, we obtain an effective translation from the labeled calculus $\LKt + \Pl$ to the nested calculus $\SKT + \Pd$ via Lemma~\ref{sktdktequiv}. 

The proof of admissibility of structural rules $\Pl$ in the presence of propagation rules $\prop$ (Lemma~\ref{Permute_Prop_and_Struc}) bears some resemblance to the proof of admissibility of structural rules $\Pd$ for $\propdeep$ in the deep nested calculi of~\cite{GorPosTiu11}. There is, however, a crucial difference in our result compared to that of \cite{GorPosTiu11}. In~\cite{GorPosTiu11}, an additional admissibility result needs to be proved {\em for every path axiom extension}: the admissibility of all display rules. 
By contrast, this additional admissibility result does not need to be proved in the labeled setting---display rules 
 are all absent in the labeled calculi. This mismatch results is an interesting observation regarding Gor\'e \textit{et al.}'s translation from $\SKT + \Pd$ to $\DKT + \propdeep$. Consider the following transformations of a proof of a nested sequent in $\SKT + \Pd$ to a proof of the same sequent in $\DKT + \propdeep$: one done directly in a nested calculus, the other through a detour in the associated labeled calculus. Note that step (3) is given by~\cite[Lem.~6.14]{GorPosTiu11} and step (5) is trivial as any derivation in $\LKt + \Pl$ is a derivation in $\LKt + \Pl + \prop$.
\begin{center}
\xymatrix{
  \SKT + \Pd  \ar@{.>}[rr]^{(4)~Thm.~\ref{DisToLab}~(+~Lem.~\ref{g3ktstrucadmiss})} \ar@{->}[d]_{(1)} &  & \LKt + \Pl  \ar@{->}[d]^{(5)}  & 		\\
 \DKT + \Pd + \propdeep + \{\rf,\rp, \ctr, \wk\}  \ar@{->}[d]_{(2)~Lem.~\ref{dktstrucadmiss}}    &  & \LKt + \Pl + \prop\ar@{->}[dd]^{(6)~ Lem.~\ref{internal}}	  & \\
 \DKT + \Pd + \propdeep \ar@{->}[d]_{(3)}  & & & \\
 \DKT + \propdeep  & & \LKt + \prop \ar@{.>}[ll]^{(7)~Lem.~\ref{LabToDis}}
 }
\end{center}
The direct translation from $\SKT + \Pd$ to $\DKT + \propdeep$ in~\cite{GorPosTiu11} is described on the left path in the above diagram; it starts with the trivial observation (1) that 
$\DKT + \Pd + \propdeep + \{\rf,\rp, \ctr, \wk\}$ subsumes $\SKT + \Pd$; followed by (2) the admissibility of display rules, contraction $\ctr$, and weakening $\wk$; and finally, (3) the admissibility of structural rules for path axioms. The detour through labeled calculus takes care of the display rules and the $\ctr$ and $\wk$ structural rules at step (4), where the admissibility of display rules is built into the canonical representation of nested sequents as polytrees (Corollary~\ref{isoresidcor}) and is completely independent of any extension with path axioms. This independence is not obviously observed in the transformation through the nested calculi. 
In fact, the designs of the propagation rules in the deep nested calculi in \cite{GorPosTiu11} take into account all possible interactions between display postulates and the path axioms and that leads to a proliferation of inference rules, e.g., for every propagation rule going downward in the syntax tree, there needs to be a symmetric version that propagates upward the tree. The proofs of admissibility of display rules in \cite{GorPosTiu11} in $\DKT$ and its extensions then need to consider all these cases, each of which is essentially the same. Moving to the labeled polytree sequent representation cuts the propagation rules by a half, and brings out the essence of a proof more clearly. These observations suggest that the syntax of the nested calculi is unnecessarily \emph{bureaucratic} in the sense that the syntactic structures of nested sequents obscure certain identities on proofs.\footnote{See e.g., \cite{Girard89} on the broader context of the use of the phrase ``bureaucracy'' in proof theory.}


For another demonstration of bureaucracy of nested sequent proofs (in comparison to labeled polytree sequent proofs): 
take a proof $\der$ of the nested sequent 
$
\circw \Gamma \}, \Delta. 
$
In proving admissibility of display postulates for $\DKT$, Gore \textit{et al.} applied a transformation (see the proof of Lemma 4.3 in \cite{GorPosTiu11}) to $\der$ to obtain another proof
$\der'$ of 
$
\Gamma, \circb \Delta \}.
$
Clearly $\der$ and $\der'$ are distinct proofs in any extension of $\DKT$, as they have distinct end sequents. But it can be shown that they both map to the same proof in the polytree representation (i.e., by simply replacing $\diam_1$ and 
$\diam_2$ rules in $\DKT$ with $\diam$ rule in labeled sequent calculus, and $\blacklozenge_1$ and $\blacklozenge_2$ 
with $\blacklozenge$). The distinction in the nested sequent proofs $\der$ and $\der'$ arises from the choice of which node in the nested sequent tree is to be designated as the root node; in the polytree representation this distinction does not arise, as there is no special node to be designated as the root node.


Let us now define the labeled propagation rules.
\begin{definition}[Propagation Graph of a Labeled Sequent] Let $\R,\Gamma$ be a labeled sequent where $N$ is the set of labels occurring in the sequent. We define the \emph{propagation graph $PG(\R,\Gamma) = ( N,E )$} to be the directed graph with the set of nodes $N$ and where $E$ is a set of labeled edges that are labeled with either a $\Diamond$ or $\blacklozenge$ as follows: For every $Rxy \in \R$, we have a labeled edge $(x,y,\Diamond) \in E$ and a labeled edge $(y,x,\blacklozenge) \in E$. Given that $PG(\R,\Gamma) = (N,E)$, we will often write $x \in PG(\R,\Gamma)$ to mean $x \in N$, and $(x,y,\Diamond) \in PG(\R,\Gamma)$ or $(y,x,\blacklozenge) \in PG(\R,\Gamma)$ to mean $(x,y,\Diamond) \in E$ or $(y,x,\blacklozenge) \in E$, respectively.
\end{definition}

\begin{definition}[Labeled Propagation Rules] Let $P$ be a set of path axioms. The set of propagation rules $\prop$ contains all rules of the form:

\begin{center}
\AxiomC{$\R, x: \ques A, y:A, \Gamma$} \RightLabel{$\propr$}
\UnaryInfC{$\R, x: \ques A, \Gamma$}
\DisplayProof
\end{center}

\noindent
where there is a path $\pi$ from $x$ to $y$ in the propagation graph of the premise and $\Pi A \rightarrow \ques A \in (P \cup I(P))^{*}$ with $\Pi$ the string of $\pi$.\footnote{Note that \emph{path} and \emph{string} are defined the same here as for nested sequents.}

\end{definition}

We now aim to prove that we can effectively transform any derivation in $\LKt + \Pl + \prop$ into a derivation in $\LKt + \prop$. This inevitably yields an effective transformation from proofs in $\LKt + \Pl$ to proofs in $\LKt + \prop$ (and eventually to $\SKT + \Pd$) in the following way: Given a derivation in $\LKt + \Pl$, we show that we can permute the topmost inference of a labeled structural rule $\negp$ upwards into the initial sequents to eliminate the use of the rule. This provides a proof in $\LKt + \Pl + \prop$ since the $\prop$ rules may be used in the permutation process to simulate the eliminated $\Pl$ rule. By permuting away all labeled structural rules $\negp \in \Pl$ from the derivation, we effectively obtain a proof in $\LKt + \prop$, which (as we will show) contains exclusively labeled polytree sequents when the end sequent is a labeled polytree sequent. The last thing that we will show in this section is how to effectively translate $\LKt + \prop$ derivations into $\DKT + \propdeep$ derivations; this result, in conjunction with Lemma \ref{sktdktequiv}, gives the desired effective translation and result.

\begin{lemma}\label{Ant_Path_Correspondence_Struc} For any structural rule $\negp$ defined relative to a path axiom $\Pi A \rightarrow \ques A$:

\begin{center}
\AxiomC{$\R, \R_{\Pi}xy, R_{\ques}xy, \Gamma$}
\RightLabel{$\negp$}
\UnaryInfC{$\R, \R_{\Pi}xy, \Gamma$}
\DisplayProof
\end{center}

\noindent
there exists a path $\pi$ in $PG(\R, \R_{\Pi}xy, R_{\ques}xy, \Gamma)$ from $x$ to $y$ whose string is $\Pi$ as well as a path from $x$ to $y$ whose string is $\ques$.

\end{lemma}

\begin{proof} Follows from the definition of $\negp$ and the definition of a propagation graph of a labeled sequent.
\end{proof}

We now show that every structural rule $\negp$ can be permuted above a propagation rule $\propr$. This lemma is vital for transforming any $\LKt + \Pl$ derivation into a $\LKt + \prop$ derivation, as it shows that all structural rules can be eliminated from a given derivation, so long as propagation rules are properly introduced. Before moving on to prove this fact however, we recall notation that was given in the introduction to Section 2. To differentiate between different path axioms in the following lemma, and make associations of paths explicit with path axioms, we use the notation $\dia{G}$ and $\dia{F}$ (potentially annotated) to represent either a white diamond $\Diamond$ or black diamond $\bldia$, and forgo the use of $\ques$ for the scope of the following lemma.

\begin{lemma}\label{Permute_Prop_and_Struc} Let $\Pset$ be a set of path axioms, $\negp \in \Pl$, $\propr \in \propl$, and $\R_{\Pi}uv := R_{\dia{G_{1}}}uz_{1}, \ldots, R_{\dia{G_{n}}}z_{n}v$. Suppose we are given a derivation that ends with:

\begin{center}
\AxiomC{$\R, \R_{\Pi}uv, R_{\dia{G}}uv, x: \dia{F} A, y:A, \Gamma$} \RightLabel{$\propr$}
\UnaryInfC{$\R, \R_{\Pi}uv, R_{\dia{G}}uv, x: \dia{F} A, \Gamma$}
\RightLabel{$\negp$}
\UnaryInfC{$\R, \R_{\Pi}uv, x: \dia{F} A, \Gamma$}
\DisplayProof
\end{center}

\noindent
where $\R_{\Pi}uv$ is active in the $\negp$ inference. Then, there exists a propagation rule $\propr' \in \propl$ such that the $\negp$ rule may be permuted upwards followed by an instance of $\propr'$ to derive the same end sequent:

\begin{center}
\AxiomC{$\R, \R_{\Pi}uv, R_{\dia{G}}uv, x: \dia{F} A, y:A, \Gamma$} \RightLabel{$\negp$}
\UnaryInfC{$\R, \R_{\Pi}uv, x: \dia{F} A, y: A, \Gamma$}
\RightLabel{$\propr'$}
\UnaryInfC{$\R, \R_{\Pi}uv, x: \dia{F} A, \Gamma$}
\DisplayProof
\end{center}

Note that 
 $\propr$ and $\negp$ may correspond to different path axioms.
\end{lemma}

\begin{proof} Suppose we are given a derivation ending with a $\propr$ inference followed by a $\negp$ inference and let $\R' = \R, R_{\dia{G_{1}}}uz_{1}, \ldots, R_{\dia{G_{n}}}z_{n}v$. Moreover, due to the application of $\propr$, there exists a path $\pi'$ of the form $x, \dia{F_{1}}, \ldots, \dia{F_{n}}, y$ from $x$ to $y$ in $PG(\R', R_{\dia{G}}uv, x: \dia{F} A, y:A, \Gamma)$. In the case where the relational atom $R_{\dia{G}}uv$ principal in $\negp$ \emph{does not} lay along the path $\pi'$ used in applying $\propr$, the two rules may be freely permuted since there is no interaction between the two. 

We therefore assume that the relational atom $R_{\dia{G}}uv$ lies along the path $\pi'$ from $x$ to $y$. By this assumption, we know that there exists an axiom $F = \dia{F_{1}} \cdots \dia{F_{m}} A \rightarrow \dia{F} A$ = $\Pi' A \rightarrow \dia{F} A \in (P \cup I(P))^{*}$, where $\Pi' = \dia{F_{1}} \cdots \dia{F_{m}}$ is the string of the path $\pi'$, corresponding to the application of $\propr$. Moreover, by our assumption that $\negp$ deletes the relational atom $R_{\dia{G}}uv$ that occurs along the path $\pi'$, the structural rule $\negp$ corresponds to a path axiom $G = \dia{G_{1}} \cdots \dia{G_{n}} A \rightarrow \dia{G} A$ where $\dia{G}$ = $\dia{F_{i}}$ for some $i \in \{1, \ldots ,m\}$. To prove the claim we must show that there exists a path $\sigma$ from $x$ to $y$ in $PG(\R', x: \dia{F} A, y:A, \Gamma)$ such that $\Sigma p \rightarrow \dia{F} p \in (P \cup I(P))^{*}$ with $\Sigma$ the string of the path $\sigma$. We construct the path $\sigma$ 
as follows: (i) replace each $u, \dia{G}, v$ in $\pi'$ with $u, \dia{G_{1}}, z_{1}, \ldots, z_{n}, \dia{G_{n}}, v$, and (ii) replace each $v, \dia{G}^{-1}, u$ in $\pi'$ with $v, \dia{G_{n}}^{-1}, z_{n}, \ldots, z_{1}, \dia{G_{1}}^{-1}, u$. Taking $\Sigma
$ to be the string of $\sigma$, we know that $\Sigma A \rightarrow \dia{F} A \in (P \cup I(P))^{*}$ since the operations performed in steps (i) and (ii) above correspond to compositions of the axioms $G$ and $I(G)$ with $F$. Let $\propr'$ be the propagation rule corresponding to the path axiom $\Sigma A \rightarrow \dia{F} A$. Since the path $\sigma$ only relies on relational atoms in $\R'$, the rule $\propr'$ may be applied after $\negp$.
\end{proof}

\begin{example} We give an example of permuting a structural rule $\negp$ above a propagation rule $\propr$. Let $\Pset := \{F,G\}$ with $F := \Diamond \blacklozenge A \rightarrow \Diamond A$ and $G  := \blacklozenge \Diamond \Diamond A \rightarrow \blacklozenge A$, where our propagation rules are defined relative to $(P \cup I(P))^{*}$. Let the application of $\propr$ correspond to the axiom $F$ and the application of $\negp$ correspond to $G$. Our derivation is given below left with the propagation graph of the initial sequent below right:

\begin{center}
\begin{tabular}{cc}
\vspace*{1 em}
\ \\
\AxiomC{}
\RightLabel{$\id$}
\UnaryInfC{$Rxv,Rxz,Rzy,Ryv,x:\Diamond p, y:p,y:\overline{p}$}
\RightLabel{$\propr$}
\UnaryInfC{$Rxv,Rxz,Rzy,Ryv,x:\Diamond p, y:\overline{p}$}
\RightLabel{$\negp$}
\UnaryInfC{$Rxv,Rxz,Rzy,x:\Diamond p, y:\overline{p}$}
\DisplayProof
&
\begin{footnotesize}
\xymatrix{
		&  \overset{\boxed{\Diamond p}}{x}\ar@{->}[dd]\ar@{->}[rrr] \ar@/^-1pc/@{.>}[dd]|-{\Diamond} \ar@/^1pc/@{.>}[rrr]|-{\Diamond} 	&   &   & \overset{\boxed{\emptyset}}{v}\ar@/^1pc/@{.>}[dd]|-{\blacklozenge} \ar@/^1pc/@{.>}[lll]|-{\blacklozenge}	&   &	\\
		& & & & & & \\
		& \overset{\boxed{\emptyset}}{z}\ar@{->}[rrr] \ar@/^-1pc/@{.>}[uu]|-{\blacklozenge}\ar@/^-1pc/@{.>}[rrr]|-{\Diamond}  &   	&	& \overset{\boxed{p, \overline{p}}}{y}\ar@/^1pc/@{.>}[uu]|-{\Diamond}\ar@/^-1pc/@{.>}[lll]|-{\blacklozenge}\ar@{->}[uu] &   &
}
\end{footnotesize}
\end{tabular}
\end{center}
The $\propr$ rule is applicable to the top sequent above because of the path $x$,$\Diamond$,$v$,$\blacklozenge$,$y$ whose string is $\Diamond \blacklozenge$, which occurs in the antecedent of $F$. However, we can see that the structural rule $\negp$ deletes the relational atom $Ryv$ that gives rise to this path. If we were to apply the $\negp$ rule first (as shown below left), the conclusion would have the propagation graph shown below right:

\begin{center}
\begin{tabular}{cc}
\raisebox{-2em}{
\AxiomC{}
\RightLabel{$\id$}
\UnaryInfC{$Rxv,Rxz,Rzy,Ryv,x:\Diamond p, y:p,y:\overline{p}$}
\RightLabel{$\negp$}
\UnaryInfC{$Rxv,Rxz,Rzy,x:\Diamond p, y:p,y:\overline{p}$}
\DisplayProof
}
&
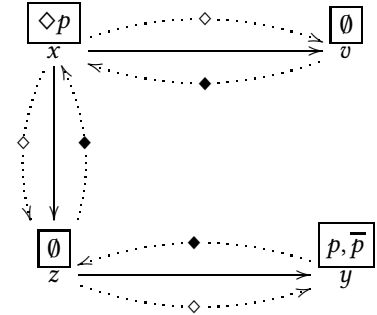
\xymatrix{
		&  \overset{\boxed{\Diamond p}}{x}\ar@{->}[dd]\ar@{->}[rrr] \ar@/^-1pc/@{.>}[dd]|-{\Diamond} \ar@/^1pc/@{.>}[rrr]|-{\Diamond} 	&   &   & \overset{\boxed{\emptyset}}{v}\ar@/^1pc/@{.>}[lll]|-{\blacklozenge}	&   &	\\
		& & & & & & \\
		& \overset{\boxed{\emptyset}}{z}\ar@{->}[rrr] \ar@/^-1pc/@{.>}[uu]|-{\blacklozenge}\ar@/^-1pc/@{.>}[rrr]|-{\Diamond}  &   	&	& \overset{\boxed{p, \overline{p}}}{y}\ar@/^-1pc/@{.>}[lll]|-{\blacklozenge} &   &
}
\end{tabular}
\end{center}

%

We construct a new path from $x$ to $y$ following the procedure explained in Lemma \ref{Permute_Prop_and_Struc} by replacing $v, \blacklozenge, y$ with $v, \blacklozenge, x, \Diamond, z, \Diamond, y$ to obtain the path $x, \Diamond, v, \blacklozenge, x, \Diamond, z, \Diamond, y$. Observe that the axiom $G \triangleright^{2} F = \Diamond \blacklozenge \Diamond \Diamond A \rightarrow \Diamond A$ is an element of the completion $(\Pset \cup I(\Pset))^{*}$. Thus, there exists a propagation rule $\propr'$ corresponding to $\Diamond \blacklozenge \Diamond \Diamond A \rightarrow \Diamond A$ which may be applied to the end sequent above to obtain the desired conclusion.

\end{example}

\begin{lemma}\label{internal} Every derivation in $\LKt + \Pl +\prop$ is effectively relatable to a derivation in $\LKt + \prop$.
\end{lemma}

\begin{proof} We prove the result by induction on the height of the given derivation in $\LKt + \Pl +\prop$; we consider the topmost application of $\negp \in \Pl$ (the general result where there are $n$ rules of $\Pl$ in our derivation is immediately obtained by applying the given procedure and successively deleting the topmost occurrences). That the output derivation is effectively relatable to the input is an immediate consequence of the rule-by-rule transformations that we use.

\textbf{Base case.} Suppose the rule $\negp$ is used on an axiom in $\LKt + \Pl +\prop$:

\begin{center}
\AxiomC{}
\RightLabel{$\id$}
\UnaryInfC{$\R, \R_{\Pi}xy, Rxy, z:p, z:\overline{p}, \Gamma $}
\RightLabel{$\negp$}
\UnaryInfC{$\R, \R_{\Pi}xy, z:p, z:\overline{p}, \Gamma $}
\DisplayProof
\end{center}

Then, it is easy to see that the conclusion is an axiom as well regardless of whether $z = x$, $z = y$, or $x \neq z \neq y$.

\textbf{Inductive step.} We show that $\lp \in \Pl$ can be permuted upward with each rule in $\LKt + \prop$:\\

(i) Permuting $\disr$ with $\negp$:\\

\begin{center}
\begin{tabular}{c c}

\AxiomC{$\R, \R_{\Pi}xy, Rxy,, z:A,z:B, \Gamma $}
\RightLabel{$\disr$}
\UnaryInfC{$\R, \R_{\Pi}xy, Rxy, z:A\lor B, \Gamma$}
\RightLabel{$\negp$}
\UnaryInfC{$\R, \R_{\Pi}xy, z:A\lor B, \Gamma$}
\DisplayProof

&

\AxiomC{$\R, \R_{\Pi}xy, Rxy, z:A,z:B, \Gamma $}
\RightLabel{$\negp$}
\UnaryInfC{$\R, \R_{\Pi}xy, z:A,z:B,  \Gamma$}
\RightLabel{$\disr$}
\UnaryInfC{$\R, \R_{\Pi}xy, z:A\lor B, \Gamma$}
\DisplayProof

\end{tabular}
\end{center}

(ii) Permuting $\conr$ with $\negp$:\\

\begin{center}
\AxiomC{$\R, \R_{\Pi}xy, Rxy, x:A, \Gamma$}
\AxiomC{$\R, \R_{\Pi}xy, Rxy, x:B, \Gamma$}
\RightLabel{$\conr$}
\BinaryInfC{$\R, \R_{\Pi}xy, Rxy, x:A\land B, \Gamma$}
\RightLabel{$\negp$}
\UnaryInfC{$\R, \R_{\Pi}xy, x:A\land B, \Gamma$}
\DisplayProof
\end{center}

\begin{center}
\AxiomC{$\R, \R_{\Pi}xy, Rxy, z:A, \Gamma$}
\RightLabel{$\negp$}
\UnaryInfC{$\R, \R_{\Pi}xy, z:A, \Gamma$}
\AxiomC{$\R, \R_{\Pi}xy, Rxy, z:B, \Gamma$}
\RightLabel{$\negp$}
\UnaryInfC{$\R, \R_{\Pi}xy, z:B, \Gamma$}
\RightLabel{$\conr$}
\BinaryInfC{$\R, \R_{\Pi}xy, z:A\land B, \Gamma$}
\DisplayProof
\end{center}

(iii) Permuting $\bboxr$ with $\negp$:\\

\begin{center}
\begin{tabular}{c c}
\AxiomC{$\R, \R_{\Pi}xy, Rxy, Rvu, v:A, \Gamma$}
\RightLabel{$\bboxr$}
\UnaryInfC{$\R, \R_{\Pi}xy, Rxy, u:\blacksquare A, \Gamma$}
\RightLabel{$\negp$}
\UnaryInfC{$\R, \R_{\Pi}xy, u:\blacksquare A, \Gamma$}
\DisplayProof

&

\AxiomC{$\R, \R_{\Pi}xy, Rxy, Rvu, v:A, \Gamma$}
\RightLabel{$\negp$}
\UnaryInfC{$\R, \R_{\Pi}xy, Rvu, v:A, \Gamma$}
\RightLabel{$\bboxr$}
\UnaryInfC{$\R, \R_{\Pi}xy, u:\blacksquare A, \Gamma$}
\DisplayProof

\end{tabular}
\end{center}

(iv) Permuting $\boxr$ with $\negp$:

\begin{center}
\begin{tabular}{c c}

\AxiomC{$\R, \R_{\Pi}xy, Rxy, Ruv, v:A, \Gamma$}
\RightLabel{$\boxr$}
\UnaryInfC{$\R, \R_{\Pi}xy, Rxy, u:\Box A, \Gamma$}
\RightLabel{$\negp$}
\UnaryInfC{$\R, \R_{\Pi}xy, u:\Box A, \Gamma$}
\DisplayProof

&

\AxiomC{$\R, \R_{\Pi}xy, Rxy, Ruv, v:A, \Gamma$}
\RightLabel{$\negp$}
\UnaryInfC{$\R, \R_{\Pi}xy, Ruv, v:A, \Gamma$}
\RightLabel{$\boxr$}
\UnaryInfC{$\R, \R_{\Pi}xy, u:\Box A, \Gamma$}
\DisplayProof
\end{tabular}
\end{center}

(v) Permuting $\bdiar$ with $\negp$: We consider the case where $Rxy$ is used in both rules; the other cases are easily shown.\\

\begin{center}
\begin{tabular}{c c}
\AxiomC{$\R, \R_{\Pi}xy, Rxy, x:A, y:\bldia A, \Gamma$}
\RightLabel{$\bdiar$}
\UnaryInfC{$\R, \R_{\Pi}xy, Rxy,y:\bldia A, \Gamma$}
\RightLabel{$\negp$}
\UnaryInfC{$\R, \R_{\Pi}xy, y:\bldia A, \Gamma$}
\DisplayProof

&

\AxiomC{$\R, \R_{\Pi}xy, Rxy, x:A, y:\bldia A, \Gamma$}
\RightLabel{$\negp$}
\UnaryInfC{$\R, \R_{\Pi}xy, x:A, y:\bldia A, \Gamma$}
\RightLabel{$\propr$}
\UnaryInfC{$\R, \R_{\Pi}xy, y:\bldia A, \Gamma$}
\DisplayProof
\end{tabular}
\end{center}

(vi) Permuting $\diar$ with $\negp$: Similar to the last case we only consider when $Rxy$ is used in both rules.

\begin{center}
\begin{tabular}{c c}

\AxiomC{$\R, \R_{\Pi}xy, Rxy, y:A, x:\Diamond A, \Gamma$}
\RightLabel{$\diar$}
\UnaryInfC{$\R, \R_{\Pi}xy, Rxy, x:\Diamond A, \Gamma$}
\RightLabel{$\negp$}
\UnaryInfC{$\R, \R_{\Pi}xy, x:\Diamond A, \Gamma$}
\DisplayProof

&

\AxiomC{$\R, \R_{\Pi}xy, Rxy, y:A, x:\Diamond A, \Gamma$}
\RightLabel{$\negp$}
\UnaryInfC{$\R, \R_{\Pi}xy, y:A, x:\Diamond A, \Gamma$}
\RightLabel{$\propr$}
\UnaryInfC{$\R, \R_{\Pi}xy, x:\Diamond A, \Gamma$}
\DisplayProof
\end{tabular}
\end{center}

(vii) Permuting $\propr$ with $\negp$: Follows from Lemma \ref{Permute_Prop_and_Struc}.
\end{proof}



The previous lemma demonstrates that propagation rules are sufficient to eliminate each structural rule $\negp$ from a given derivation of $\LKt + \Pl$. The consequence of this elimination procedure is that each derivation of $\LKt + \Pl$ may be effectively transformed into a derivation in $\LKt + \prop$. Our goal now is to prove that every derivation in $\LKt + \prop$ (of a labeled polytree sequent) contains nothing but labeled polytree sequents, which implies that each labeled sequent in such a derivation is in the domain of our translation function $\UTGd$. This fact proves that the effective transformation of the previous lemma has gotten us closer to completing the full translation of proofs from $\LKt + \Pl$ to $\SKT + \Pd$ since our derivations contain labeled sequents that are readily convertible to nested sequents via the function $\UTGd$.

\begin{lemma}\label{Lpath2} Let $P$ be a set of path axioms. Every $\LKt + \prop$ proof of a labeled polytree sequent consists solely of labeled polytree sequents.
\end{lemma}

\begin{proof} Similar to the proof of Lemma \ref{Lab-equiv}. Observe that all rules of $\LKt + \prop$ preserve disconnectivity and cycles downward in an inference.
\end{proof}

Before showing that every derivation in $\LKt + \prop$ is translatable to a derivation in $\DKT + \propdeep$, we must ensure that applications of propagation rules can be properly translated from the labeled to the deep nested setting. The following lemma proves that every path in a labeled polytree sequent is present in its translatee under $\UTGd$, that is, translating a labeled polytree sequent into a nested sequent preserves paths. We use this fact (Lemma~\ref{Applicability_LabProp_Implies_DKTProp}) in the sequel to effectively translate instances of propagation rules.

\begin{lemma}\label{Applicability_LabProp_Implies_DKTProp}
For any labeled polytree sequent $\R, \Gamma$ with a path $\pi$ from a label $x$ to a label $y$ in its propagation graph, the path $\pi$ exists in $PG(\UTGd_{z}(\R, \Gamma))$ from $x$ to $y$ (where $z$ is an arbitrary node in $\R, \Gamma$).
\end{lemma}

\begin{proof} Let $\R, \Gamma$ be a labeled polytree sequent with a path $\pi$ from $x$ to $y$ in its propagation graph. We translate $\R, \Gamma$ to a nested sequent relative to the node $x$ (as opposed to an arbitrary node $z$) and let the nodes in $PG(\UTGd_{x}(\R, \Gamma))$ be the same as those in the given labeled polytree sequent. Note that by Lemma \ref{DE_For_Diff_Node}, translating $\R, \Gamma$ relative to any label yields a display equivalent sequent, and by Lemma \ref{DisEquiv_Sequents_Identical_PG} the propagation graphs of all such sequents are identical. Therefore, the claim will hold regardless of the node chosen to translate from, meaning we are permitted to translate from $x$. We now prove the claim by induction on the length of the path connecting $x$ and $y$.

\textbf{Base case.} The case when the path from $x$ to $y$ is of length $0$ is trivial, so we show the case when the path from $x$ to $y$ is of length $1$. Suppose that there is a forward edge from $x$ to $y$, that is, $\pi = x,\Diamond,y$ (the case when there is a backward edge from $x$ to $y$ is similar). Then, $\UTGd_{x}(\R, \Gamma)$ will be a nested sequent with a $\circ$-edge from $x$ to $y$, and so the labeled edge $(x,y,\Diamond)$ is in the propagation graph. 


\textbf{Inductive step.} Suppose there is a path $x,...,z,\ques,y$ from $x$ to $y$ of length $n+1$. Therefore, there is a path of length $n$ from $x$ to $z$, and a path of length $1$ from $z$ to $y$ in $PG(\R, \Gamma)$. By the inductive hypothesis, the path from $x$ to $z$ occurs in $\UTGd_{x}(\R, \Gamma)$. By the base case, the path $z, \ques, y$ also occurs in $\UTGd_{x}(\R, \Gamma)$. Therefore, the path $x,...,z,\ques,y$ is in $\UTGd_{x}(\R, \Gamma)$.
\end{proof}

\begin{lemma}\label{G3Kt-Path-to-DKT-Path} Every derivation of a sequent $\R,\Gamma$ in $\LKt + \prop$ consisting solely of labeled polytree sequents is effectively relatable to a derivation of $\UTGd(\R,\Gamma)$ in $\DKT + \propdeep$.
\end{lemma}

\begin{proof} We extend the proof of Lemma \ref{G3KT_to_DKT} and include the inductive case for translating propagation inferences.

If we assume that a labeled propagation rule is used last in the given derivation, then there must be a corresponding axiom $\Pi p \rightarrow \ques p \in (P \cup I(P))^{*}$ whose antecedent allows for an application of the rule. This axiom will also define a deep nested propagation rule:

\begin{center}
\begin{tabular}{c@{\hskip 1cm}c}
\AxiomC{$\R, x:\ques A, y:A, \Gamma$}
\RightLabel{$\propr$}
\UnaryInfC{$\R, x:\ques A, \Gamma$}
\DisplayProof

&

\AxiomC{$X[\ques A]_{x}[A]_{y}$}
\RightLabel{$\propr$}
\UnaryInfC{$X[\ques A]_{x}[\emptyset]_{y}$}
\DisplayProof
\end{tabular}
\end{center}

By Lemma \ref{Applicability_LabProp_Implies_DKTProp}, the propagation rule may be applied in the deep nested proof because the path $\pi$ from $x$ to $y$ (whose string is $\Pi$) exists in the propagation graph of the premise $\UTGd(\R, x:\ques A, y:A, \Gamma)$ = $X[\ques A]_{x}[A]_{y}$.
\end{proof}

\begin{lemma}\label{LabToDis}
Every derivation of a labeled polytree sequent $\R,\Gamma$ in $\LKt + \propl$ is effectively relatable to a derivation of $\UTGd(\R,\Gamma)$ in $\DKT + \propdeep$.
\end{lemma}

\begin{proof} Let $\der$ be our derivation of $\R,\Gamma$ in $\LKt + \propl$. By Lemma~\ref{Lpath2}, we know that every sequent occurring in $\der$ will be a labeled polytree sequent. By the previous lemma, we may effectively translate this derivation into a derivation in $\DKT + \propdeep$.
\end{proof}

\begin{theorem}\label{reverse-translation-theorem} Every derivation of a labeled polytree sequent $\R,\Gamma$ in $\LKt + \Pl$ is effectively relatable to a derivation of $\UTGd(\R,\Gamma)$ in $\SKT + \Pd$.
\end{theorem}

\begin{proof} By Lemma~\ref{internal} we know that every derivation $\der$ of a labeled polytree sequent $\R,\Gamma$ in $\LKt + \Pl$ is effectively transformable to a derivation $\der'$ of $\R,\Gamma$ in $\LKt + \prop$. By Lemma \ref{LabToDis}, there is an effective translation of $\der'$ to a proof $\der''$ of $\UTGd(\R,\Gamma)$ in $\DKT + \propdeep$. Lemma \ref{sktdktequiv} implies that we can effectively transform $\der''$ in $\DKT + \propdeep$ into a derivation $\der'''$ of $\UTGd(\R,\Gamma)$ in $\SKT + \Pd$. The composition of effective procedures gives an effective procedure, which gives the result.
\end{proof}

Note that the application of Lemma~\ref{sktdktequiv} in the above theorem is a rather heavy proof-theoretic transformation since it invokes cut-elimination. Nevertheless, the output derivation is still effectively relatable since cut-elimination is defined locally, at the level of the proof rules.
    
\section{Concluding Remarks}\label{conclusion}

One consequence of our work is a methodology for proving the conservativity of shallow nested (i.e. display) calculi under the deletion of certain logical rules. For example, if $\SKT+NestSt(A \rightarrow \Diamond A)$ is a (sound and complete) shallow nested calculus for the logic $\Kt + A \rightarrow \Diamond A$, is $\SKT - \{(\bBox), (\blacklozenge)\}+NestSt(A\rightarrow \Diamond A)$ a (sound and complete) shallow nested calculus for $\K + A\rightarrow \Diamond A$? Notice that a derivation in the latter calculus may contain a sequent with the structural connective~$\bullet \{ \cdot \}$ (which could be introduced via the residuation rule $\rf$) or  even though the corresponding logical connective~$\bBox$ is not an operator in the ($\bBox$,$\blacklozenge$-free) language of $\K + A \rightarrow \Diamond A$ (meaning that a sequent such as $\circ \{\bullet \{p\}, \bullet \{q\}\}$ cannot be interpreted as a formula). Therefore, care must be taken when attempting to identify the logic obtained under the deletion of logical rules for connectives $\heart_{1}, \ldots, \heart_{n}$, since structural connectives that act as proxies for $\heart_{1}, \ldots, \heart_{n}$ will still be present in sequents and therefore may give the calculus increased expressive power.

A general solution which establishes the conservativity of  display calculi for tense logics over their modal
fragments, by making use of algebraic semantics, has been presented in~\cite{GreMaPalTziZha16}. Our work obtains this result \emph{syntactically} in the context of tense logics with \emph{modal general path axioms}
by exploiting the translations developed in the previous sections (Corollary~\ref{cor:modal-translation}). This subsumes the conservativity results in~\cite{GorPosTiu11}, for the more restricted set of \emph{modal path axioms}.

Another interesting consequence of our work is the suggestion of a potential methodology for constructing labeled calculi suitable for proof-search
and for proving decidability of the associated logics. The labeled calculus formalism offers a uniform method for obtaining cut-, contraction-, and weakening-admissible calculi for a large class of logics \cite{Neg05,Sim94,Vig00}. The drawback of such calculi is that they contain structural rules which are not immediately well-suited for proof-search; if the rules are applied na\"ively bottom-up, then proof-search may not terminate. Therefore, auxiliary results concerning a bound on the number of times a rule needs to be applied is required to ensure termination, see, e.g. \cite{Neg05}. Nevertheless, the method presented here of internalizing labeled calculi for path extensions of~$\Kt$ shows that such structural rules can be eliminated from a labeled derivation in the presence of appropriate, auxiliary inference rules.
This opens up an avenue for future research and gives rise to new questions: for what other logics can labeled structural rules be eliminated in favor of rules better adapted for proof-search? Is there an effective procedure for determining such rules? Note that this procedure has been investigated in~\cite{LyoBer19,Lyo21} and has shown that the method of refining labeled calculi is applicable to a variety of logics.

Moreover, the obtained internalized labeled calculi lend themselves nicely to uniformly proving interpolation for the class of path extensions of~$\Kt$~\cite{LyoTiuGorClo19}. As explained in Section~\ref{Trans_Path_Ax_Ext_Subsection}, labeled polytrees provide a canonical representation of nested sequents that encode the polytree structure in the set $\R$ of relational atoms, and the decorations of the nodes as the labeled formulae in $\Gamma$. Such a representation makes it easier to define a generalized notion of interpolant, and to observe useful relationships between such interpolants (e.g. a generalized notion of duality via the $\cut$ rule)~\cite{LyoTiuGorClo19}.

The relationship between Kripke frames and the algebraic semantics for modal logics is well-studied (see e.g.~\cite{BlaRijVen01}). Because labeled calculi are based on the former, and shallow nested (display) calculi on the latter, the bi-directional translations established in this work can be interpreted as demonstrating this relationship concretely, at the level of an inference rule.

Lastly, let us briefly discuss the issues with translating labeled proofs for all general path extensions of $\Kt$ into shallow nested proofs. As explained in Section 5, when translating from labeled to shallow nested, we first pre-process our labeled derivation by eliminating all \emph{path} structural rules from the given derivation, trading such inferences for propagation rules. This methodology does not immediately appear applicable to the class of \emph{general path} structural rules, since it is not clear what propagation rules (if any) should be added to a labeled calculus to allow for complete structural rule elimination. Therefore, it appears that an alternative methodology is required to carry out the translation of labeled proofs for all general path extensions of $\Kt$, which we defer to future work.

\begin{acks}
This project has received funding from the \grantsponsor{1}{European Union's Horizon 2020 research and innovation programme}{} under the Marie Sk{\l}odowska-Curie grant agreement No \grantnum[]{1}{689176}. Work additionally supported by the \grantsponsor{2}{FWF}{} projects: \grantnum[]{2}{START Y544-N23}, \grantnum[]{2}{I 2982}, \grantnum[]{2}{W1255-N23}, and \grantnum[]{2}{P 33548}.
R. Ramanayake would like to acknowledge the financial support of the CogniGron research center and the Ubbo Emmius Funds (Univ. of Groningen).
We are grateful to the anonymous reviewers for their careful reading and constructive comments.
\end{acks}


\bibliographystyle{ACM-Reference-Format}
\bibliography{mybib}

\end{document}